\documentclass[aps,pra,twocolumn,superscriptaddress,floatfix,nofootinbib,showpacs,longbibliography]{revtex4-1}

\usepackage[utf8]{inputenc}  
\usepackage[T1]{fontenc}     
\usepackage[table]{xcolor}
\usepackage[british]{babel}  
\usepackage[sc,osf]{mathpazo}
\usepackage[scaled=0.86]{berasans}  
\usepackage[colorlinks=true, citecolor=blue, urlcolor=blue]{hyperref}  
\usepackage{graphicx} 
\usepackage[babel]{microtype}  
\usepackage{amsmath,amssymb,amsthm,bm,amsfonts,mathrsfs,bbm} 

\usepackage{xspace}  
\usepackage{pgf,tikz}
\usepackage{xcolor}
\usepackage{multirow}
\usepackage{array}
\usepackage{bigstrut}
\usepackage{braket}
\usepackage{color}
\usepackage{natbib}
\usepackage{multirow}
\usepackage{mathtools}
\usepackage{float}
\usepackage[caption = false]{subfig}
\usepackage{xcolor,colortbl}
\usepackage{color}
\usepackage{rotating}
\usepackage{tikz}
\usepackage{ulem}
\usetikzlibrary{quantikz}

\newcommand{\be}{\begin{equation}}
\newcommand{\ee}{\end{equation}}
\newcommand{\ba}{\begin{eqnarray}}
\newcommand{\ea}{\end{eqnarray}}
\newcommand{\ketbra}[2]{|#1\rangle \langle #2|}
\newcommand{\tr}{\operatorname{Tr}}

\newtheorem{theorem}{Theorem}

\newtheorem{proposition}{Proposition}

\def\>{\rangle}
\def\<{\langle}

\usepackage{centernot}
\usepackage{subfig}

\begin{document}

\title{Certifying beyond quantumness of locally quantum no-signalling theories through quantum input Bell test}

\author{ Edwin Peter Lobo}
\affiliation{School of Physics, IISER Thiruvananthapuram, Vithura, Kerala 695551, India.}

\author{Sahil Gopalkrishna Naik}
\affiliation{Department of Physics of Complex Systems, S.N. Bose National Center for Basic Sciences, Block JD, Sector III, Salt Lake, Kolkata 700106, India.}	

\author{Samrat Sen}
\affiliation{Department of Physics of Complex Systems, S.N. Bose National Center for Basic Sciences, Block JD, Sector III, Salt Lake, Kolkata 700106, India.}

\author{Ram Krishna Patra}
\affiliation{Department of Physics of Complex Systems, S.N. Bose National Center for Basic Sciences, Block JD, Sector III, Salt Lake, Kolkata 700106, India.}

\author{Manik Banik}
\affiliation{Department of Physics of Complex Systems, S.N. Bose National Center for Basic Sciences, Block JD, Sector III, Salt Lake, Kolkata 700106, India.}

\author{Mir Alimuddin}
\affiliation{Department of Physics of Complex Systems, S.N. Bose National Center for Basic Sciences, Block JD, Sector III, Salt Lake, Kolkata 700106, India.}	

\begin{abstract}
Physical theories constrained with local quantum structure and satisfying the no-signalling principle can allow beyond-quantum global states. In a standard Bell experiment, correlations obtained from any such beyond-quantum bipartite state can always be reproduced by quantum states and measurements, suggesting local quantum structure and no-signalling to be the axioms to isolate quantum correlations. In this letter, however, we show that if the Bell experiment is generalized to allow local quantum inputs, then beyond-quantum correlations can be generated by every beyond-quantum state. This gives us a way to certify beyond-quantumness of locally quantum no-signalling theories and in turn suggests requirement of additional information principles along with local quantum structure and no-signalling principle to isolate quantum correlations. More importantly, our work establishes that the additional principle(s) must be sensitive to the quantum signature of local inputs. We also generalize our results to multipartite locally quantum no-signalling theories and further analyze some interesting implications. 
\end{abstract}



\maketitle

{\it Introduction.--} Correlations among distant events established through the violation of Bell type inequalities confirm nonlocal behavior of the physical world \cite{Bell64,Bell66,Mermin93,Brunner14}. Nonseparable multipartite quantum states yielding such correlations, in Schr\"{o}dinger's words, are ``...the characteristic trait of quantum mechanics, the one that enforces its entire departure from classical lines of thought" \cite{Schrodinger35}. The advent of quantum information science identifies the power of such nonlocal correlations in numerous device independent protocols -- cryptographic key distribution \cite{Barrett05,Acin06,Acin07}, randomness certification \cite{Pironio10} and amplification \cite{Colbeck12}, dimension witness \cite{Brunner08,Gallego10,Mukherjee15} are few canonical examples. Cirel'son's result \cite{Cirelson80}, however, establishes that the nonlocal strength of quantum correlations is limited compared to the general {\it no-signalling} (NS) ones \cite{Popescu94} as depicted in the celebrated Clauser-Horne-Shimony-Holt (CHSH) inequality violation \cite{Clauser69}.
	
To comprehend the limited nonlocal behavior of quantum theory and to obtain a better understanding of the theory itself, researchers have proposed several approaches to compare and contrast quantum theory with other conceivable physical theories constructed within more general mathematical frameworks \cite{Birkhoff36,Mackey63,Ludwi68,Mielni68,Beltramett81,Soler95,Haag96,Clifton03,Barrett07,Abramsky08,Chiribella11}. Here, we consider a class of theories wherein local measurements are described quantum mechanically, but they allow global structure more generic than quantum theory \cite{Foulis80,Klay87,Wallach00,Barnum05,Barnum10,Acin10,Torre12,Kleinmann13}. Gleason-Busch celebrated result in quantum foundations proves that any map from generalized measurements to probability distributions can be written as the trace rule with the appropriate quantum state \cite{Gleason57,Busch03} (see also \cite{Caves04}). This theorem, when appraised to the case of local observables acting on multipartite  systems, hence called the unentangled Gleason's theorem, endorses the joint NS probability distributions to be obtained from some Hermitian operator called the positive over all pure tensors (POPT) state \cite{Foulis80,Klay87,Wallach00,Barnum05}. Although the set of POPT states is strictly larger than the set of quantum states (density operators), in a recent work, Barnum {\it et al.} have shown that the set of bipartite correlations attainable from the POPT states is precisely the set of quantum correlations \cite{Barnum10}. Consequently, their result provokes a far-reaching conclusion "... that if nonlocal correlations beyond quantum mechanics are obtained in any experiment then quantum theory would be invalidated even locally." 

In this letter, we analyze the correlations of multipartite POPT states obtained from local measurements performed on their constituent parts by considering a generalized Bell scenario as introduced in \cite{Buscemi12}. While in the standard Bell scenario spatially separated parties receive some classical inputs and accordingly generate some classical outputs by performing local measurements on their respective parts of some composite system, recently Buscemi has generalized the scenario where the parties receive quantum inputs instead of classical variables \cite{Buscemi12}. In this generalized scenario he has shown that all entangled states exhibit nonlocality, despite some of them allowing local-hidden-variable (LHV) model in classical input scenario \cite{Werner89,Barrett02,Rai12}. Considering this generalized scenario, here we show that not all correlations obtained from bipartite POPT states are quantum simulable. In fact, every beyond quantum POPT state produces some beyond quantum correlation in some quantum input game. On the other hand, to illustrate the limitations of the standard Bell scenario, we show that there are POPT states which produce classical-input-classical-output correlations that are not only quantum simulable, rather simulable classically. Our result shows that the {\it strong} claim made by the authors in \cite{Barnum10} will not be correct anymore in this generalized Bell scenario which, as we will show, is allowed within the framework of local quantum theory. From a foundational perspective our study welcomes new information principles incorporating this generalized Bell type scenario to isolate quantum correlation from beyond-quantum ones. We also analyze the implication of this generalized scenario for multipartite correlations and answer an open question raised in \cite{Acin10}.

{\it Gleason's theorem.--} We investigate the class of locally quantum theories studied in a series of works in the recent past \cite{Foulis80,Klay87,Wallach00,Barnum05,Barnum10,Acin10,Torre12,Kleinmann13}. In accordance with these works, we say that Alice is {\it locally quantum} if her physical system is described by a Hilbert space $\mathcal{H}_A$ with dimension $d_A$ and her measurements $M_A$ are given by a collection of effects corresponding to positive-operator-valued measurement (POVM) \cite{Kraus83} operators $\{\pi^a_A\}_a$ acting on $\mathcal{H}_A$ and satisfying the constraint $\sum_a\pi^a_A=\mathbb{I}_A$; where $\forall~a,~\pi^a_A\in\mathcal{E}(\mathcal{H}_A)\subset\mathcal{L}(\mathcal{H}_A)$, with $\mathcal{E}(\mathcal{H}_A)$ and $\mathcal{L}(\mathcal{H}_A)$ respectively denoting the set of all positive operators and bounded linear operators acting on $\mathcal{H}_A$; and $\mathbb{I}_A$ is the identity operator on  $\mathcal{H}_A$. The probability $p(a|M_A)$ that Alice obtains an outcome $a$ for measurement $M_A\equiv\{\pi^a_A\}_a$ is given by a  generalized probability measure $\mu:\mathcal{E}(\mathcal{H}_A)\mapsto[0,1]$, satisfying the properties (i) $\forall~\pi^a_A\in\mathcal{E}(\mathcal{H}_A),~0\le\mu(\pi^a_A)\le1,~$ (ii) $\mu(\mathbb{I}_A)=1$, and (iii) $\mu(\sum_i\pi^i_A)=\sum_i\mu(\pi^i_A)$ for any sequence $\pi^1_A,\pi^2_A\cdots$ with $\sum_i\pi^i_A\le \mathbb{I}_A$. Each probability measure $\mu$ corresponds to a `state' in the local quantum theory. We can make the association with the familiar quantum theory in which states are described by density operators by invoking the Gleason-Busch theorem according to which any such generalized probability measure is given by a linear functional of the form $\mu(\pi^a_A)=\tr(\rho_A\pi^a_A)$, for some density operator $\rho_A\in\mathcal{D}(\mathcal{H}_A)$; $\mathcal{D}(\mathcal{H}_A)$ denotes the set of positive operators with unit-trace on $\mathcal{H}_A$.

Interesting situations arise when the theorem is generalized to the case of local observables acting on multipartite systems. Each party is assumed to be locally quantum as described above, with the $i^{th}$ party performing the measurement $M_{A_i} \equiv\{\pi^a_{A_i}\}_a $. The `state' is now given by a probability measure $\mu: \times_{i=1}^n \mathcal{E}(\mathcal{H}_{A_i}) \mapsto [0,1]$. According to Unentangled Gleason's theorem \cite{Foulis80,Klay87,Wallach00,Barnum05}, any such functional $\mu$ satisfying the no-signalling condition is of the form $\mu(\pi^{a_1}_{A_1}, \cdots , \pi^{a_n}_{A_n})=\tr[W(\pi^{a_1}_{A_1} \otimes \cdots \otimes \pi^{a_n}_{A_n})]$, where $W$ is a Hermitian, unit trace operator. Thus, the `states' of multipartite locally quantum theory are in one-to-one correspondence with the operators $W$. $W$, being positive over all pure tensors, is called a POPT state. However, positivity of $W$ over entangled effects is not assured and such a non-positive $W$ can act as an entanglement witness operator \cite{Guhne09}. The set of POPT states $\mathcal{W}(\bigotimes_i\mathcal{H}_{A_i})$ includes $\mathcal{D}(\bigotimes_i\mathcal{H}_{A_i})$ as a proper subset and a $W$ will be called `beyond quantum state' (BQS) whenever $W\in\mathcal{W}(\bigotimes_i\mathcal{H}_{A_i})\setminus\mathcal{D}(\bigotimes_i\mathcal{H}_{A_i})$. With an aim to study the correlations obtained from BQSs we briefly recall the standard Bell scenario.  

{\it Standard Bell scenario.--} A multipartite Bell scenario can be described as the following Prover-Verifier task. $n$ distant Verifiers $A_1, A_2, \cdots, A_n$ have their own source of classical indices $s_i\in\mathcal{S}_i$. With the aim to verify some global property of a composite state prepared by a powerful but untrustworthy Prover, they send their respective indices as inputs to spatially separated subsystems of the composite systems. Classical outputs $a_i\in\mathcal{O}_i$ are generated from the different subsystems of the composite system and accordingly some payoff $\mathscr{P}:\times_{i=1}^n(\mathcal{S}_i\times\mathcal{O}_i)\mapsto\mathbb{R}$ is calculated. An implicit rule is that no communication is allowed among different subsystems once the game starts. Upon playing the game sufficiently many times, the input-output correlation $P\equiv\{p(a_1\cdots a_n|s_1\cdots s_n)\}^{a_i\in\mathcal{O}_i}_{s_i\in\mathcal{S}_i}$ is obtained. The collection of all NS correlations forms a convex polytope $\mathcal{NS}$. A correlation is called classical if and only if it is of the form $p_L(a_1\cdots a_n|s_1\cdots s_n)=\int_\Lambda p(\lambda)\prod_ip(a_i|s_i,\lambda)d\lambda$, where $\lambda\in\Lambda$ is some classical variable shared among the parties. Collection of such correlations also forms a convex polytope $\mathcal{L}$. On the other hand, a correlation is called quantum if it is obtained from some quantum state through local measurements, {\it i.e.} $p_Q(a_1\cdots a_n|s_1\cdots s_n)=\tr[\rho(\bigotimes_i\pi^{a_i}_{s_i})]$ for some $\pi^{a_i}_{s_i}\in\mathcal{E}(\mathcal{H}_{A_i})$ and $\rho\in\mathcal{D}(\bigotimes_i\mathcal{H}_{A_i})$. The set of all quantum correlations $\mathcal{Q}$ forms a convex set but not a polytope. The framework of locally quantum theories allows us to define the correlation set obtained from the POPT states. Following the terminology of Ref. \cite{Acin10} we call such a correlation 'Gleason correlation' and denote the set as $\mathcal{GL}$. The following set inclusion relations have been established: $\mathcal{L}\subsetneq\mathcal{Q}\subseteq\mathcal{GL}\subsetneq\mathcal{NS}$. While the first proper inclusion follows from the seminal work of Bell \cite{Bell64}, the last one is due to Cirel'son and Popescu-Rohrlich \cite{Cirelson80,Popescu94}. On the other hand the equality $\mathcal{Q}=\mathcal{GL}$ for bipartite correlations is established in \cite{Barnum10}. More precisely, the authors in \cite{Barnum10} have shown that for every POPT $W_{AB}$ and for every local measurements $M_A=\{\pi^a_A\}_a$ and $M_B=\{\pi^b_B\}_b$, there exists a quantum state $\rho_{AB} \in \mathcal{D}(\mathcal{H}_A \otimes \mathcal{H}_B)$ and measurements $\tilde{M}_A=\{\tilde{\pi}^a_A\}_a$, $\tilde{M}_B=\{\tilde{\pi}^b_B\}_b$ such that, $\tr[W_{AB}(\pi^a_A\otimes \pi^b_B)]=\tr[\rho_{AB}(\tilde{\pi}^a_A\otimes \tilde{\pi}^b_B)]$.
In this classical input-output scenario we are now in a position to prove our first result that in some sense can be considered stronger than the result of Barnum {\it et al.}
\begin{proposition}\label{prop1}
There exist beyond quantum bipartite states yielding correlations that are  classically simulable. 
\end{proposition}
\begin{proof}
(Sketch) The family of operators $W_p:=p\Gamma[\ket{\phi^+}\bra{\phi^+}]+(1-p)\mathbb{I}/4$ is a BQS for $1/3< p\le1$; $\ket{\phi^+}:=(\ket{00}+\ket{11})/\sqrt{2}$ and $\Gamma$ denotes partial transposition. If we consider projective measurements only then a LHV description is possible whenever $p\le1/2$, whereas for generic POVMs one can have such a description for $p\le5/12$. The LHV models are motivated from the well known constructions of Werner \cite{Werner89} and Barrett \cite{Barrett02}. The explicit construction we defer to the Supplementary part. 
\end{proof}
The result of Barnum {\it et al.} \cite{Barnum10} and our Proposition \ref{prop1} depicts the limitation of classical-input classical-output Bell scenario to reveal the full correlation strength of BQSs. At this point a more general Bell scenario turns out to be advantageous.

{\it Semiquantum Bell scenario.--} 
\begin{figure}[b!]
\centering
\includegraphics[width=0.5\textwidth]{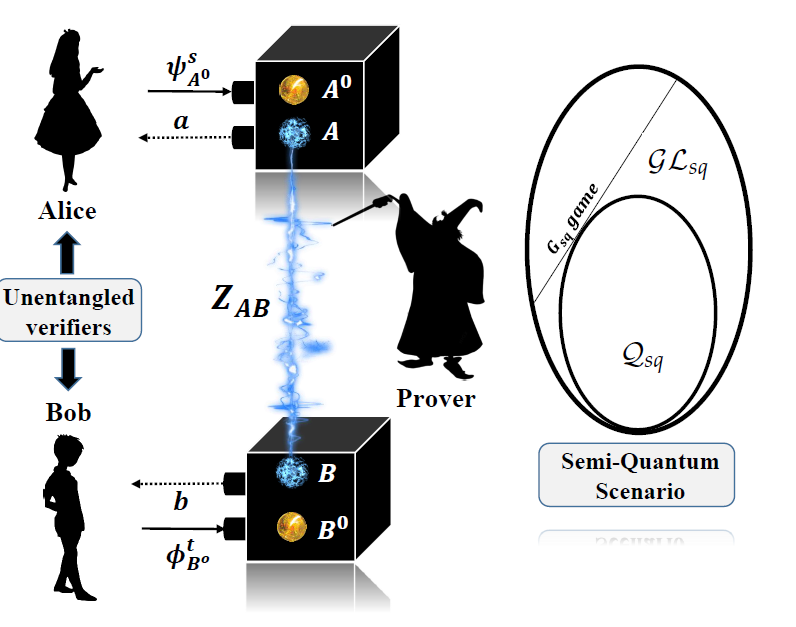}
\caption{(Color online) A powerful but untrustworthy Prover distributes a bipartite state $Z_{AB}$ between two distant Verifiers (Alice and Bob). The Verifiers do not have any entanglement between them, but possess their own trusted local quantum preparation device. Such limited resourceful Verifiers can verify the beyond quantumness of the state $Z_{AB}$ provided to them (Theorem \ref{theo1}). The seminal Hahn-Banach separation theorem plays a crucial role in making this verification possible -- the correlations produced from the bipartite quantum states form a convex-compact proper subset within the set of correlations produced from all bipartite states compatible with local quantum description and NS principle.}\label{fig1}
\end{figure}
The scenario was introduced by Buschemi to establish the nonlocal behaviour of all entangled quantum states \cite{Buscemi12}, which has subsequently generated a plethora of research interests \cite{Branciard13,Banik13,Chaturvedi15,Rosset20,Schmid20,Graffitti20}. In this scenario, each of the Verifiers, assumed to be {\it locally quantum,} has a random source of pure quantum states $\{\{p_i(s_i),\psi^{s_i}_{A_i^o}\}|s_i\in\mathcal{S}_i\}$ (see Fig. \ref{fig1}). They wish to verify whether the state of a global system $W_{A_1 \cdots A_n}$, prepared by a powerful but untrustworthy Prover, is BQS or not.
To this aim they provide their respective quantum states to the different parts of the distributed global state. The Prover returns some classical index $a_i\in\mathcal{O}_i$ by performing local quantum measurements $M_{A_iA_i^o} = \{ \pi^{a_i}_{A_iA_i^o}\}_{a_i}$ on the respective distributed parts of the global state and the states received from the Verifiers. Accordingly, some payoff $\beta:\times_{i=1}^n(\mathcal{S}_i\times\mathcal{O}_i)\mapsto\mathbb{R}$ is given, which specifies a semi-quantum game $\mathbb{G}_{sq}$. From the global state $W_{A_1\cdots A_n}$, the Prover can generate a correlation $P_{W_{A_1\cdots A_n}}:=\{p(a_1,\cdots,a_n|\psi^{s_1},\cdots,\psi^{s_n})\}$ and the expected payoff is calculated as $\mathcal{I}_{\mathbb{G}_{sq}}(W_{A_1\cdots A_n}):=\sum_{s_1,a_1,\cdots, s_n, a_n}\beta\left(s_1,a_1,\cdots,s_n,a_n\right)\times p\left(a_1,\cdots,a_n|\psi^{s_1},\cdots,\psi^{s_n}\right)$. Like the standard scenario, we can define the set of correlations $\mathcal{X}_{sq}$ with $\mathcal{X}\in\{\mathcal{L},\mathcal{Q},\mathcal{GL},\mathcal{NS}\}$ and $\mathcal{X}\subseteq\mathcal{X}_{sq}$ in general. When the quantum sources consist of orthogonal quantum states, the scenario boils down to standard Bell scenario and no distinction is possible between a bipartite entangled state and a BQS \cite{Barnum10}. 

{\it Composing POPT states.--} In the Semiquantum Bell scenario, the Prover performs local measurements $\{\pi^{a_i}_{A_i A_i^o}\}_{a_i}$ on the $i^{th}$ subsystem. The composite multipartite state is given by a functional $\mu_{A_1A_1^o\cdots A_nA_n^o}$ which, invoking unentangled Gleason's theorem, corresponds to a POPT state $Z_{A_1A_1^o\cdots A_nA_n^o}$. The form of $Z_{A_1A_1^o\cdots A_nA_n^o}$ must be consistent with the states held by the Verifiers and the Prover. If the states held by the Verifiers are pure and unentangled, then one can show that $Z_{A_1A_1^o\cdots A_nA_n^o} = \otimes_i\psi^{s_i}_{A_i^o}\otimes W_{A_1\cdots A_n}$. We leave the details to the supplemental part. Interestingly, our next result shows that within local quantum description the unentangled Verifiers (hence weakly resourceful) can test the property `entangled vs BQS' supplied by the more resourceful Prover. 
\begin{theorem}\label{theo1}
For every beyond quantum state $W_{AB}\in\mathcal{W}(\mathcal{H}_A\otimes\mathcal{H}_B)$ there exists a semiquantum game $\mathbb{G}_{sq}$ such that $\mathcal{I}_{\mathbb{G}_{sq}}(W_{AB})<0$, while $\mathcal{I}_{\mathbb{G}_{sq}}(\rho_{AB})\geq0,~\forall~\rho_{AB}\in\mathcal{D}(\mathcal{H}_A\otimes\mathcal{H}_B)$.
\end{theorem}  
\begin{proof}
At the core of our proof lies the classic Hahn-Banach separation theorem of convex analysis and the fact that for every beyond quantum state $W_{AB}\in\mathcal{W}(\mathcal{H}_A\otimes\mathcal{H}_B)$ there exists an entangled state $\chi_{AB}\in\mathcal{D}(\mathcal{H}_A\otimes\mathcal{H}_B)$ such that $\tr[W_{AB}\chi_{AB}]<0$, whereas $\tr[\sigma_{AB}\chi_{AB}]\geq 0$, $\forall~\sigma_{AB}\in\mathcal{D}(\mathcal{H}_A\otimes\mathcal{H}_B)$ \cite{Woronowicz76,Horodecki96,Stormer13}. Also note that, there exits (non-unique) choices of pure states $\psi^s_A \in\mathcal{D}(\mathcal{H}_A)~\&~\psi^t_B\in\mathcal{D}(\mathcal{H}_B)$, and some real coefficients $\{\beta_{s,t}\}$ such that $\chi_{AB}= \sum_{s,t}\beta_{s,t} {\psi^s_A}^\mathrm{T} \otimes {\psi^t_B}^\mathrm{T}$ where $T$ represents the transposition with respect to the computational basis. This leads us to the required game $\mathbb{G}_{sq}^\chi$ where the verifiers Alice and Bob yield quantum inputs $\psi^s_{A^o}$ and $\phi_{B^o}^t$, and ask the Prover to return outputs $\in \{0,1\}$ from the distributed parts of the global state. The average payoff is calculated as $\mathcal{I}:=\sum_{s,t}\beta_{s,t}p(11|\psi^s_{A^o}\psi^t_{B^o})$. The measurement $\{P^{+}_{uu^o},\mathbb{I}_{uu^o}-P^{+}_{uu^o}\}$ is performed on the distributed parts of the global state, where $P_{uu^o}^{+}:=\ket{\phi^+}_{uu^o}\bra{\phi^+}$ with $\ket{\phi^+}_{uu^o}:=\frac{1}{\sqrt{d_u}}\sum^{d_u-1}_{i=0}|ii\rangle$ and $P^{+}_{uu^o}$ corresponds to the outcome  $1$, $u\in\{A,B\}$. We therefore have,
\begin{align*}
\mathcal{I}_{\mathbb{G}^\chi_{sq}}\left(W_{AB}\right)&=\sum_{s,t}\beta_{s,t}\tr\left[P^{+}_{AA^o}\otimes P^{+}_{BB^o}\left(\psi^s_{A^o}\otimes W_{AB} \otimes \psi^t_{B^o}\right)\right] \\
&=\sum_{s,t}\beta_{s,t}\tr\left[\left(R_A\otimes R_B\right)W_{AB}\right];
\end{align*} 
where $R_A$ and $R_B$ are the effective POVMs acting on the parts of Alice's and Bob's shares of the BQS, respectively; and are given by $R_u:=\tr_{u^o}[P^+_{uu^o}(\mathbb{I}_u\otimes \psi^s_{u^o})]=\frac{1}{d_u}{\psi^s_{u}}^\mathrm{T}$. Therefore, we have
\begin{align*}
\mathcal{I}_{\mathbb{G}^\chi_{sq}}\left(W_{AB}\right)&=\frac{1}{d_Bd_A}\sum_{s,t}\beta_{s,t}\tr\left[\left({\psi^s_{A}}^\mathrm{T} \otimes {\psi^t_{B}}^\mathrm{T}\right)W_{AB}\right]\\
&=\frac{1}{d_Bd_A}\tr\left[\left(\sum_{s,t}\beta_{s,t}{\psi^s_{A}}^\mathrm{T} \otimes {\psi^t_{B}}^\mathrm{T}\right)W_{AB}\right]\\
&=\frac{1}{d_Bd_A}\tr\left[\chi_{AB}W_{AB}\right]<0.
\end{align*}
On the other hand, given an arbitrary quantum state $\rho_{AB}$ let the measurements $M_{AA^o}\equiv\{\pi^a_{AA^o}\})_a$ and $N_{BB^o}=\{\pi^b_{BB^o}\}_b$ be performed, where $a,b \in\{0,1\}$. The average payoff turns out to be 
\begin{align*}
\mathcal{I}_{\mathbb{G}^\chi_{sq}}\left(\rho_{AB}\right)&=\sum_{s,t}\beta_{s,t}\tr\left[\pi^1_{AA^o}\otimes \pi^1_{BB^o}\left(\psi^s_{A^o}\otimes \rho_{AB} \otimes \psi^t_{B^o}\right)\right] \\
&=\sum_{s,t}\beta_{s,t}\tr\left[R_{A^oB^o}\left(\psi^s_{A^o}\otimes \psi^t_{B^o}\right)\right],
\end{align*} 
where $R_{A^oB^o}:=\tr_{AB}[(\pi^1_{AA^o}\otimes \pi^1_{BB^o})(\mathbb{I}_{A^oB^o}\otimes \rho_{AB})]$ is a positive operator. Using linearity of trace we get,
\begin{align*}
\mathcal{I}_{\mathbb{G}^\chi_{sq}}\left(\rho_{AB}\right)&=\tr\left[R_{A^oB^o}\left(\sum_{s,t}\beta_{s,t}\psi^s_{A^o}\otimes \psi^t_{B^o}\right)\right]\\
&=\tr\left[R_{A^oB^o}~\chi^\mathrm{T}_{A^oB^o}\right]\geq0.
\end{align*}
The last inequality follows due to the fact that $\chi^\mathrm{T}_{A^oB^o}$ is a valid density operator, and this completes the proof.
\end{proof}
Theorem \ref{theo1} establishes that $\mathcal{Q}_{sq}\subsetneq\mathcal{GL}_{sq}$ in the bipartite scenario. Note that, following an argument similar to \cite{Banik13}, it can be shown that in the semiquantum scenario, even if classical communication between different distributed parts are allowed to the Prover along with the quantum entangled state $\rho_{AB}$, still the local statistics obtained from BQS cannot be simulated.
Our result poses some interesting questions. The proper set inclusion relation $\mathcal{Q}\subsetneq\mathcal{NS}$ established in \cite{Popescu94} has motivated several novel approaches to isolate quantum correlations from beyond-quantum ones \cite{vanDam05,Buhrman10,Pawlowski09,Navascues09,Fritz13,Cabello13,Oppenheim10,Banik13(1),Banik15,Kar16(1),Kar16(2),Banik19,Bhattacharya20}. Along similar lines, the proper set inclusion relation $\mathcal{Q}_{sq}\subsetneq\mathcal{GL}_{sq}$ welcomes new principle(s) to isolate the quantum correlations from beyond-quantum ones in this generalized scenario. Importantly, our Theorem \ref{theo1} suggests that such principles must be sensitive to the quantum signature of local inputs \cite{Rosset20,Schmid20}.

The semi-quantum scenario also has important implications while studying correlations in multipartite (involving more than two parties) scenarios. Acín {\it et al.} have already pointed out that the result of Barnum {\it et al.} does not generalize to the multipartite scenario even in the classical-input classical-output paradigm \cite{Acin10}. They have provided examples of multipartite BQSs producing beyond quantum correlations within the standard Bell scenario. They have also pointed out that a BQS of the form
\begin{align}\label{qpopt}
W_{A_1\cdots A_N}&=\sum_kp_k(\Lambda^k_{A_1}\otimes \cdots \otimes \Lambda^k_{A_N})[\rho^k_{A_1\cdots A_n}],
\end{align}
will not generate any classical-input classical-output correlation that lies outside the set of correlations generated by quantum states. Here, $\{p_k\}$ is a probability distribution, $\rho^k_{A_1\cdots A_n}\in\mathcal{D}(\bigotimes_i\mathcal{H}_{A_i})$, and $\Lambda_i^k$ are positive but not completely positive trace preserving maps on $\mathcal{L}(\mathcal{H}_{A_i})$ \cite{Stormer13}. The authors in \cite{Acin10} have left the question open to identify the additional requirements to close the gap in their result. Our next result provides a solution to close this gap.  
\begin{theorem}\label{theo2}
For every BQS $W_{A_1\cdots A_N}\in\mathcal{W}(\bigotimes_{i=1}^N\mathcal{H}_{A_i})$ there exists a semiquantum game $\mathbb{G}_{sq}$ such that $\mathcal{I}_{\mathbb{G}_{sq}}(W_{A_1\cdots A_N})<0$, whereas $\mathcal{I}_{\mathbb{G}_{sq}}(\rho_{A_1\cdots A_N})\geq0,~\forall\rho_{A_1\cdots A_N}\in\mathcal{D}(\bigotimes_{i=1}^N\mathcal{H}_{A_i})$.
\end{theorem}
The proof is a straightforward generalization of the proof of Theorem \ref{theo1} (see supplemental). While Theorem \ref{theo1} \& \ref{theo2} are just existence theorems, it is not hard to see that given an arbitrary BQS there is an efficient algorithm to construct a semiquantum game (the procedure is discussed in supplemental). It is important to note that non-orthogonal quantum inputs are necessary to reveal the beyond quantum signature of correlation for any BQS of the form of Eq.(\ref{qpopt}). This implicitly follows from the results of Barnum {\it et al.} \cite{Barnum10} and Acín {\it et el.} \cite{Acin10}. It is worth mentioning that this semi-quantum scenario is different from local tomography as it establishes beyond quantum nature of POPT states in a measurement device independent manner where the measurement devices used to produce the classical outcomes need not be trusted \cite{Branciard13}.

{\it Discussion:} One of the earnest research endeavours in quantum theory is to understand the limited nonlocal behaviour of quantum correlations. Apart from the foundational appeal, this question  also has practical relevance as nonlocal correlations have been established as useful resources for several tasks. In the bipartite scenario the result of Barnum {\it et al.} \cite{Barnum10} provides an answer to this question by assuming the description of  local systems to be quantum. Our work, however, points out the limitation of the scenario considered in \cite{Barnum10}. The authors there have not considered the most general bipartite scenario allowed within the unentangled Gleason-Busch framework, which assumes local quantum measurement and the no-signalling principle. Within this framework, the  type of inputs allowed are not restricted to classical indices, rather they can be quantum states. Our Theorem \ref{theo1} shows that all bipartite beyond quantum states compatible with unentangled Gleason-Busch theorem can yield beyond quantum correlations in the quantum input scenario, and accordingly divulges a more complex picture within the correlations zoo. Our study therefore welcomes new principle(s) to isolate the correlations obtained from quantum states, and more importantly, suggests that such a principle should take the type of inputs into consideration as indistinguishability of nonorthogonal quantum input states plays a crucial role in making the distinction between Quantum and BQS states. 

Our Theorem \ref{theo2} establishes that within the quantum input paradigm all multipartite BQSs yield beyond quantum correlations which was known earlier only for a particular class of such states \cite{Acin10}. After the work of \cite{Acin10}, Torre {\it et al.} have shown that when the local systems are identical qubits, any theory admitting at least one continuous reversible interaction must be identical to quantum theory \cite{Torre12}. However, the result in \cite{Torre12} has also been obtained within the classical-input classical-output paradigm. It might be interesting to see what additional structures are required there to single out the quantum correlations in the quantum-input scenario. On the other hand, within the classical-input-classical-output paradigm, the authors in \cite{DallArno17} and the present authors with other collaborators in \cite{Naik21,Sen2022} have studied beyond quantum correlations in the time-like domain. Similar studies with quantum inputs might provide new insights there. 

\begin{acknowledgements}
We would like to thank Tamal Guha, Guruprasad Kar, Francesco Buscemi, and Markus P. Müller for their comments on the preliminary version of our manuscript. MA and MB acknowledge funding from the National Mission in Interdisciplinary Cyber-Physical systems from the Department of Science and Technology through the I-HUB Quantum Technology Foundation (Grant no: I-HUB/PDF/2021-22/008). MB acknowledges support through the research grant of INSPIRE Faculty fellowship from the Department of Science and Technology, Government of India and the start-up research grant from SERB, Department of Science and Technology (Grant no: SRG/2021/000267).
\end{acknowledgements}

\onecolumngrid
\section*{Supplemental}
\section{Proof of Proposition ${\bf 1}$}
\begin{proof}
The POPTness of the state directly follows from the expression $\tr[W_p(\mathbb{P}_{\hat{n}}\otimes\mathbb{P}_{\hat{m}})]=1/4(1+p~\hat{n}.\hat{m})$, where $\mathbb{P}_{\hat{x}}:=1/2(\mathbb{I}+\hat{x}.\sigma)$; and beyond quantumness follows from the explicit eigenvalue calculation of the operator $W_p$.

Our aim is to show that for certain range of the parameter $p$, the classical-input classical-output correlations obtained from the class of BQSs  $W_p:=p\Gamma[\ket{\phi^+}\bra{\phi^+}]+(1-p)\mathbb{I}/4$ can be classically simulated. Given the classical inputs, the parties Alice and Bob perform some local measurement on their part of the BQS to obtained some classical outputs. The joint input-output probabilities are calculated using Born rule as the local systems are assumed to be quantum. By classically simulable we mean that the obtained correlations allow a local hidden variable (LHV) model, {\it i.e.} if Alice and Bob perform some measurements $A\equiv\{A_i~|~A_i\ge0~\&~\sum_iA_i=\mathbb{I}\}$ and $B\equiv\{B_j~|~B_j\ge0~\&~\sum_jB_j=\mathbb{I}\}$ respectively, then the joint probability distributions are factorizable.
\begin{align}\label{P1}
P(A_i,B_j|A,B,W_p)=\int_{\Lambda}\omega(\lambda|W_p)P(A_i|A,\lambda)P(B_j|B,\lambda)d\lambda,  
\end{align}
where $\lambda\in\Lambda$ is some shared variable (also called common cause/HV) and $\omega(\lambda|W_p)$ is a probability distribution on the HV space $\Lambda$. 

Let us first consider the particular case, where measurement effects of Alice's and Bob's measurements are proportional to some rank one projection operator, {\it i.e.} $A_i=x_iP_i~\&~B_j=y_jQ_j$, with $0< x_i,y_j\leq1$. Note that, $\tr\left[\Gamma[\ket{\phi^+}\bra{\phi^+}]\left(P_i\otimes Q_j\right)\right]=1/4(1+\hat{m}_i.\hat{n}_j)$, where $P_i:=1/2(\mathbb{I}+\hat{m}_i.\sigma)$ and $Q_j:=1/2(\mathbb{I}+\hat{n}_j.\sigma)$. This expression differs from $\tr\left[\ket{\psi^-}\bra{\psi^-}\left(P_i\otimes Q_j\right)\right] =1/4(1-\hat{m}_i.\hat{n}_j)$ just by a negative sign, which motivates us to construct the LHV for $W_p$ by simply modifying the LHV model known for the noisy singlet state \cite{Barrett02}.

Let the hidden variables $\lambda\in\Lambda$ (which we shall now denote by $\ket{\lambda}$) be the unit vectors of a $2$-dimensional complex Hilbert space. The local responses are given by,
\begin{align}\label{P2}
 P(A_i|A,\lambda)&:=\bra{\lambda}A_i\ket{\lambda}\Theta\left(\bra{\lambda}P_i\ket{\lambda}-\frac{1}{2}\right)+\frac{x_i}{2}\left(1-\sum_k\bra{\lambda}A_k\ket{\lambda}\Theta\left(\bra{\lambda}P_i\ket{\lambda}-\frac{1}{2}\right)\right);\\ \label{P3}
 P(B_j|B,\lambda)&:= y_j\left(1-\bra{\lambda^{\perp}}Q_j\ket{\lambda^{\perp}}\right);
\end{align}
where $\Theta(x)$ is the Heaviside step function and $\ket{\lambda^\perp}$ is the state perpendicular to $\ket{\lambda}$. It is noteworthy that the response on Alice's side is contextual, since the response of the effect $A_i$ depends on the other effects $A_k$'s constituting the measurement $A$.  Substituting Eqs.(\ref{P2}) and (\ref{P3}) in Eq.(\ref{P1}) we obtain
\begin{align}
P\left(A_i,B_j|A,B,W_p\right)&=\int_{\Lambda}{\omega}\left(\lambda|W_p\right)\left[y_j\left(1-\bra{\lambda^{\perp}}Q_j\ket{\lambda^{\perp}}\right)\right]\times\left[\bra{\lambda}A_i\ket{\lambda}\Theta\left(\bra{\lambda}P_i\ket{\lambda}-\frac{1}{2}\right)\right.\nonumber\\
&~~~~~~~~~~~~~~~~~~~~~~~~~~~~~\left.+\frac{x_i}{2}\left(1-\sum_k\bra{\lambda}A_k\ket{\lambda}\Theta\left(\bra{\lambda}P_i\ket{\lambda}-\frac{1}{2}\right)\right)\right] d\lambda.\label{P4}
\end{align}

To check the reproducibility condition let us define the following quantity:
\begin{align}
J_{ij}:=x_iy_j \int d\lambda ~ \omega(\lambda|W_p)\Theta\left(\bra{\lambda}P_i\ket{\lambda}-\frac{1}{2}\right)\bra{\lambda}P_i\ket{\lambda}\bra{\lambda^{\perp}}Q_j\ket{\lambda^{\perp}}. \label{P5}    
\end{align}
 Using Eq.(\ref{P5}), Eq.(\ref{P4}) can be written as 
\begin{align}
P\left(A_i,B_j|A,B,W_p\right) &=\left(-J_{ij}-\frac{1}{2}x_iy_j\int d\lambda ~\omega\left(\lambda|W_p\right)\bra{\lambda^{\perp}}Q_j\ket{\lambda^{\perp}}\right)+\left(\frac{x_i}{2}\sum_kJ_{kj}\right)\nonumber\\
&~~~~~~~~~~~~~~~~+\left(y_j\sum_lJ_{il}\right)+\left(\frac{x_iy_j}{2}\int d\lambda ~ \omega\left(\lambda|W_p\right)\right)-\left(\frac{x_iy_j}{2}\sum_{kl}J_{kl}\right)\nonumber\\
&=\frac{x_iy_j}{2}\left(\int d\lambda ~\omega\left(\lambda|W_p\right)-\int d\lambda ~\omega\left(\lambda|W_p\right)\bra{\lambda^{\perp}}Q_j\ket{\lambda^{\perp}}\right)\nonumber\\
&~~~~~~~~~~~~~~+\left(-J_{ij}+\frac{x_i}{2}\sum_kJ_{kj}\right)+\left(y_j\sum_lJ_{il}-\frac{x_iy_j}{2}\sum_{kl}J_{kl}\right).\label{P6}
\end{align}
The quantity $c\equiv\int d\lambda ~\omega\left(\lambda|{W_p}\right)\bra{\lambda^{\perp}}Q_j\ket{\lambda^{\perp}}$ is invariant under the change of $Q_j$. Now, $\sum_jy_jc=\sum_jy_j\int d\lambda~\omega\left(\lambda|{W_p}\right)\bra{\lambda^{\perp}}Q_j\ket{\lambda^{\perp}}=\int d\lambda~\omega\left(\lambda|{W_p}\right)=1$. Thus $c=1/2$; which yields,
\begin{align}
P\left(A_i,B_j|A,B,W_p\right)=\frac{x_iy_j}{4}+\left(-J_{ij}+\frac{x_i}{2}\sum_kJ_{kj}\right)+\left(y_j\sum_lJ_{il}-\frac{x_iy_j}{2}\sum_{kl}J_{kl}\right).\label{P7}
\end{align}
In order to evaluate $J_{ij}$ we write $\ket{\lambda}=z_0\ket{0}+z_1\ket{1}$ where $\{\ket{0},\ket{1}\}$ is an orthonormal basis for $\mathbb{C}^2$. Let $z_\nu=r_{\nu}e^{i\theta_{\nu}}$ for $\nu\in\{0,1\}$. We choose $\ket{0}$ to be such that $\ket{0}\bra{0}=P_i$. Writing $u_{\nu}=r_{\nu}^2$ and $Q_j=\ket{q_j}\bra{q_j}$ and using the fact $\braket{\lambda^{\perp}|q_j}\braket{q_j|\lambda^{\perp}}=\braket{\lambda|q_j^{\perp}}\braket{q_j^{\perp}|\lambda}$, we get
\begin{align}
J_{ij}&=x_iy_j\int d\lambda ~ \omega\left(\lambda|W_p\right)\Theta(\bra{\lambda}P_i\ket{\lambda}-1/2)\bra{\lambda}P_i\ket{\lambda}\braket{\lambda|q_j^{\perp}}\braket{q_j^{\perp}|\lambda}\nonumber\\
&=\frac{1}{N}x_iy_j\sum_{\nu=0}^1|\braket{q_j^{\perp}|\nu}|^2\int_{1/2}^1du_0\int_0^1du_1~\delta(u_0+u_1-1)~u_0u_{\nu}=x_iy_j\sum_{\nu=0}^1~|\braket{q_j^{\perp}|\nu}|^2~J_{\nu},\label{P8}
\end{align}
where we have assumed $\omega\left(\lambda|W_p\right)$ to be a uniform distribution over $\Lambda$ and
\begin{align}
N&:=\int_{0}^1du_0\int_0^1du_1~\delta(u_0+u_1-1),\\
J_{\nu}&:=\frac{1}{N}\int_{1/2}^1du_0\int_0^1du_1~\delta(u_0+u_1-1)~u_0u_{\nu}.
\end{align}
Defining $\tilde{J}\equiv\frac{1}{N}\int_{1/2}^1du_0\int_0^1du_1~\delta(u_0+u_1-1)~u_0~$ and using $u_0+u_1=1$, we can write $J_1=\tilde{J}-J_0$. Using normalization condition for $\ket{q_j^{\perp}}$ and the fact $\ket{p_i}=\ket{0}$ we get $|\braket{q_j^{\perp}|1}|^2=1-|\braket{q_j^{\perp}|p_i}|^2$,  which thus yields
\begin{align}
J_{ij}=x_iy_j\left[\left(\Tilde{J}-J_0\right)+\left(2J_0-\Tilde{J}\right) |\braket{q_j^{\perp}|p_i}|^2\right].\label{P9}
\end{align}
Substituting Eq.(\ref{P9}) in Eq.(\ref{P7}) we get,
\begin{equation}
 P\left(A_i,B_j|A,B,W_p\right)=x_iy_j\left(\frac{1+4J_0-2\Tilde{J}}{4}-(2J_0-\Tilde{J})|\braket{q_j^{\perp}|p_i}|^2\right)
     \label{P10}
\end{equation}
Straightforward integration yields $J_0=\frac{7}{24}$ and $\Tilde{J}=\frac{3}{8}$, which further implies
\begin{align}
P\left(A_i,B_j|A,B,W_p\right)&=\frac{x_iy_j}{48}\left(17-10|\braket{q_j^{\perp}|p_i}|^2\right)=\frac{x_iy_j}{48}\left(17-10\tr\left(P_iQ_j^{\perp}\right)\right)\nonumber\\
&=x_iy_j\frac{1}{4}\left(1+\frac{5}{12}\hat{m}_i\cdot\hat{n}_j\right).\label{P11}
\end{align}
Again, from the Born rule we have
\begin{equation}
 P\left(A_i,B_j|A,B,W_p\right)=\tr\left[(A_i\otimes  B_j)(W_p)\right]=x_iy_j\frac{1}{4}\left(1+p~\hat{m}_i\cdot\hat{n}_j\right).\label{P12}
\end{equation}
Therefore, for $p=\frac{5}{12}$ the BQS $W_p$ allows a LHV model when all the effects constituting Alice's and Bob's measurements are proportional to rank one projectors. We are now left to extend this model for more general measurements (consisting more than rank one effects). This can be argued by noticing that any POVM element $\mathcal{A}$ is a Hermitian operator with $0< \mathcal{A}\leq1$, and hence allows spectral decomposition of the form $\mathcal{A}=\sum_i A_i$, where $A_i = x_i P_i$ are operators proportional to rank-one projectors like the ones considered above with  $0< x_i\leq1$ and $P_{i}P_{j}=\delta_{ij}P_{i}$. Thus any general POV measurement can be regarded as a coarse-grained measurement of the special scenario considered above. We associate the outcome $A_i$ of the finer measurement with the outcome $\mathcal{A}$ of the coarse-grained measurement for all values of $i$. Thus we have a LHV model for $W_{{5}/{12}}$.

Once the LHV model is defined for a particular state, it can be extended for a large class of states. Suppose we have a LHV model for the state $\sigma_1$. It is then possible to construct a LHV model for a state $\sigma_2$ if it can be written in the form $  \sigma_2=\sum_{ij}M_i\otimes N_j~\sigma_1~M_i^{\dagger}\otimes N_j^{\dagger}$, such that $\sum_iM_i^{\dagger}M_i=\mathbb{I},~\&~\sum_jN_j^{\dagger}N_j=\mathbb{I}$. For describing the LHV model of $\sigma_2$ we just need to modify the responses in the following way
\begin{align}
P_{\sigma_2}(A_i|A,\lambda)&:= P_{\sigma_1}(A_i^{\prime}|A^{\prime},\lambda)~~\mbox{where}~~A_i^{\prime}:=\sum_kM_k^{\dagger}A_iM_k,\\
P_{\sigma_2}(B_j|B,\lambda)&:= P_{\sigma_1}(B_j^{\prime}|B^{\prime},\lambda)~~~\mbox{where}~~B_j^{\prime}:=\sum_lN_l^{\dagger}B_jN_l.
\end{align}
 If we now take $\omega\left(\lambda|\sigma_2\right)=\omega\left(\lambda|\sigma_1\right)$, the above construction will give
\begin{align}
\int d\lambda ~\omega\left(\lambda|\sigma_2\right)P_{\sigma_2}(A_i|A,\lambda)P_{\sigma_2}(B_j|B,\lambda)
&=\tr[\left(A_i^{\prime}\otimes B_j^{\prime}\right)\sigma_1]\nonumber\\
&=\sum_{kl}\tr[\left(M_k^{\dagger}A_iM_k\otimes N_l^{\dagger}B_jN_l\right)\sigma_1]\nonumber\\
&=\sum_{kl}\tr[\left(A_i\otimes B_j\right)\left(M_k\otimes N_l~\sigma_1~M_k^{\dagger}\otimes N_l^{\dagger}\right)]\nonumber\\
&=\tr[\left(A_i\otimes B_j\right)\sigma_2].
\end{align}
Thus the above construction is successful in defining a valid LHV model for $\sigma_2$. It is obvious that any $W_{p^\prime}$ can be created from $W_p$ just by using local operations if $p^\prime\leq p$. This implies existence of a LHV model for $W_{p\leq\frac{5}{12}}$. Therefore, any classical-input classical-output correlation obtained from $W_p$ is classically simulable whenever $p\le5/12$. On the other hand, as discussed in the manuscript, $W_p$ is BQS for $p>1/3$. This proves the claim of Proposition ${\bf 1}$. 
\end{proof}

{\bf Remark:} Motivated from the LHV model constructed in \cite{Werner89,Popescu94(1)}, it can be further shown that for $W_p$ a classical model exists for $p\le1/2$ whenever Alice's and Bob measurements are limited to projective measurement only. In this case also $\lambda$'s are given by unit vectors of $2$-dimensional complex Hilbert space and $\omega\left(\lambda|W_p\right)$ is taken to be uniform distribution. Using spherical polar coordinates we can denote the HVs as $\hat{\lambda}=\sin(\theta)\cos(\phi)\hat{i}+\sin(\theta)\sin(\phi)\hat{j}+\cos(\theta)\hat{k}$ and $\omega\left(\lambda|W_p\right)d\lambda=\frac{1}{4\pi}\sin(\theta)d\theta d\phi$. Alice's and Bob's response are given by
\begin{align}
 P(P_i|A,\lambda)&=\cos^2(\alpha_1/2)=\frac{1+\cos(\alpha_1)}{2};\\
 P(Q_j|B,\lambda)&=1~~~~\mbox{if}~~~2\cos^2(\alpha_2/2) < 1;\\
    &=0~~~~\mbox{if}~~~2\cos^2(\alpha_2/2) > 1.
\end{align}
Here $\alpha_1$ is the angle between the block vector of $P_i$ and $-\hat{\lambda}$, and $\alpha_2$ is the angle between the block vector of $Q_j$ and $\hat{\lambda}$. Without any loss of generality we can consider $P_i=1/2(\mathbb{I}+\hat{m}_i.\sigma)$ and $Q_j=1/2(\mathbb{I}+\hat{n}_j.\sigma)$ with $\hat{m}_i=(\sin~x,0,\cos~x)$ and $\hat{n}_j=(0,0,1)$. This implies $P(Q_j|B,\lambda)=1$ for $\frac{\pi}{2}<\alpha_2<\frac{3\pi}{2}$ and accordingly we have non zero contribution in the integral for $\frac{\pi}{2}<\theta<\pi$. Also we get $\cos(\alpha_1)=-\sin(x)\sin(\theta)\cos(\phi)-\cos(x)\cos(\theta)$. Therefore, we have
\begin{align}
P\left(P_i,Q_j|A,B,W_p\right)&=\int_{\theta=\frac{\pi}{2}}^{\theta=\pi}\int_{\phi=0}^{\phi=2\pi}\frac{1}{4\pi}\sin(\theta)\frac{1}{2}\left[1-\sin(x)\sin(\theta)\cos(\phi)-\cos(x)\cos(\theta)\right]~d\theta~d\phi\nonumber\\
&=\frac{1}{4}+\frac{\cos~x}{8}=\frac{1}{4}(1+\frac{1}{2} \hat{m}_i\cdot\hat{n}_j),
\end{align}
which is same as the Born probability obtained from the state $W_p$ for $p=1/2$. It is not difficult to see that the model can be extended for any values of $p\le1/2$.

\section{Composing POPT sates} \label{foundational implication}
In the Semiquantum Bell scenario, the Prover performs local measurements $M_{A_iA_i^o} = \{\pi^{a_i}_{A_i A_i^o}\}_{a_i}$ on the $i^{th}$ subsystem. In the locally quantum no-signalling framework, the composite multipartite state is given by a functional $\mu_{A_1A_1^o\cdots A_nA_n^o}$ which, invoking unentangled Gleason's theorem, corresponds to a POPT state $Z_{A_1A_1^o\cdots A_nA_n^o}$. If the states held by the Verifiers are pure and unentangled, then one can show that $Z_{A_1A_1^o\cdots A_nA_n^o} = \otimes_i\psi^{s_i}_{A_i^o}\otimes W_{A_1\cdots A_n}$. To begin with, it is easy to check that $\otimes_i\psi^{s_i}_{A_i^o}\otimes W_{A_1\cdots A_n}$ is a valid POPT state, as shown below: 

\begin{proposition} \label{propA1}
For every POPT state $W_{A_1 \cdots A_n}$, the tensor product state $W_{A_1\cdots A_n}\otimes\psi^{s_1}_{A_1^o}\otimes\cdots\otimes\psi^{s_n}_{A_n^o}$ is also a POPT state.
\end{proposition} 
\begin{proof}
For any set of local POVMs $ \pi_{A_1 A_1^o}\otimes \cdots \otimes \pi_{A_n A_n^o} $ acting on the tensor product state $W_{A_1\cdots A_n} \otimes \psi^{s_1}_{A_1^o} \otimes \cdots \otimes \psi^{s_n}_{A_n^o}$, we have,

\begin{align}
    \tr\left[\left( \pi_{A_1 A_1^o} \otimes  \cdots \otimes \pi_{A_n A_n^o} \right)\left( W_{A_1\cdots A_n} \otimes \psi^{s_1}_{A_1^o} \otimes \cdots \otimes \psi^{s_n}_{A_n^o} \right)\right] &= \tr\left[\left(\pi_{A_1} \otimes \cdots \otimes \pi_{A_n} \right)\left( W_{A_1\cdots A_n} \right)\right] \nonumber \\
    &\ge 0 \nonumber\\
    \text{where, } \pi_{A_j} := \tr_{A_j^o}\left[\left( \mathbb{I}_{A_j}\otimes \psi^{s_j}_{A_j^o} \right) \left(\pi_{A_j A_j^o}  \right)\right] \text{  is a positive}&\text{ operator for every POVM element  } \pi_{A_j A_j^o.} \nonumber
\end{align}
The final inequality follows from the fact that $W_{A_1\cdots A_n}$ is a POPT state and gives positive probabilities for local measurements. Therefore, $ W_{A_1\cdots A_n} \otimes \psi^{s_1}_{A_1^o} \otimes \cdots \otimes \psi^{s_n}_{A_n^o}$ produces a positive probability for all local measurements. Hence, it is a valid POPT state. 
\end{proof}
We next show that subsystems of POPT states are obtained by the partial trace operation. In what follows, we restrict the analysis to bipartite scenarios for notational convenience. The generalization to multipartite scenarios is obvious. For any local measurements $M_{AA^o}=\{\pi^{a}_{A A^o}\}_{a}$ and $M_{BB^o}=\{\pi^{b}_{B B^o}\}_{b}$ we have,
\begin{align*}
    \mu_{AA^oBB^o}(\pi^a_{AA^o}, \pi^b_{BB^o}) &= \tr[Z_{AA^oBB^o}(\pi^a_{AA^o}\otimes \pi^b_{BB^o})]
\end{align*}
Let $\pi^{a}_{A A^o} = \pi^{\alpha}_{A} \otimes \pi^{\alpha^o}_{A^o}$ and $\pi^{b}_{B B^o} = \pi^{\beta}_B \otimes \pi^{\beta^o}_{B^o}$ with $\sum_{\alpha} \pi^{\alpha}_{A} = \mathbb{1}_{A}, \sum_{\alpha^o} \pi^{\alpha^o}_{A^o} = \mathbb{1}_{A^o}, \sum_{\beta} \pi^{\beta}_B =  \mathbb{1}_B, \sum_{\beta^o} \pi^{\beta^o}_{B^o} =  \mathbb{1}_{B^o}.$
\begin{align}
    \mu_{AA^oBB^o}(\pi^{\alpha}_{A} \otimes \pi^{\alpha^o}_{A^o}, \pi^{\beta}_B \otimes \pi^{\beta^o}_{B^o}) &= \tr[Z_{AA^oBB^o}(\pi^{\alpha}_{A} \otimes \pi^{\alpha^o}_{A^o}\otimes \pi^{\beta}_B \otimes \pi^{\beta^o}_{B^o})] \label{eq22} \\
    \text{Summing over }& \alpha^o \text{ and } \beta^o, \nonumber \\
    \mu_{AB}(\pi^{\alpha}_{A}, \pi^{\beta}_B) = \tr[W_{AB}(\pi^{\alpha}_{A}\otimes \pi^{\beta}_B)] &=  \tr[Z_{AA^oBB^o}(\pi^{\alpha}_{A} \otimes \mathbb{1}_{A^o}\otimes \pi^{\beta}_B \otimes \mathbb{1}_{B^o})] \label{eq23}
\end{align}
Where, we have used $\sum_{\alpha^o,\beta^o}\mu_{AA^oBB^o}(\pi^{\alpha}_{A} \otimes \pi^{\alpha^o}_{A^o}, \pi^{\beta}_B \otimes \pi^{\beta^o}_{B^o}) = \sum_{\alpha^o,\beta^o} p(\alpha,\alpha^o,\beta,\beta^o|M_{AA^o} M_{BB^o}) = p(\alpha,\beta | M_A M_B) = \mu_{AB}(\pi^{\alpha}_{A}, \pi^{\beta}_B)$. Since equation Eq. (\ref{eq23}) is true for all POVMs $\pi^\alpha_A$ and $\pi^\beta_B$, we get,
\begin{align}
  \tr_{A^oB^o}(Z_{AA^oBB^o}) = W_{AB} \label{eq24}
\end{align}
If we sum over $\beta$ and $\beta^o$ in Eq. (\ref{eq22}), we get
\begin{align}
    \mu_{AA^o}(\pi^{\alpha}_{A}, \pi^{\alpha^o}_{A^o}) = \tr[\tr_{BB^o}(Z_{AA^oBB^o})(\pi^{\alpha}_{A}\otimes \pi^{\alpha^o}_{A^o})] \label{eq25}
\end{align}
$\tr_{BB^o}(Z_{AA^oBB^o})$ is a Positive Operator since $Z_{AA^oBB^o}$ must be a POPT in the $AA^o/BB^o$ cut. If we now sum Eq. (\ref{eq25}) over $\alpha$ we end up with,
\begin{align}
    \mu_{A^o}(\pi^{\alpha^o}_{A^o}) = \tr(\psi_{A^o} \pi^{\alpha^o}_{A^o}) = \tr[\tr_{BB^o}(Z_{AA^oBB^o})({\mathbb{1}_{A}}\otimes \pi^{\alpha^o}_{A^o})] \label{eq26}
\end{align}
where $\psi_{A^o}$ is the state given the Prover by the Verifier Alice. Since Eq. (\ref{eq26}) is true for all $\pi^{\alpha^o}_{A^o}$, we have,
\begin{align}
   \tr_{A} (\tr_{BB^o}(Z_{AA^oBB^o})) &= \psi_{A^o} \label{eq27}\\
   \tr_{A^o} (\tr_{BB^o}(Z_{AA^oBB^o})) &= \tr_B(W_{AB}) \text{ ~~ [From Eq. (\ref{eq24})}] \label{eq28}
\end{align}
Along with the fact that $\tr_{BB^o}(Z_{AA^oBB^o})$ and $\tr_B(W_{AB})$ are Positive operators, and that $\psi_{A^o}$ is a pure quantum state, equations (\ref{eq27}) and (\ref{eq28}) imply,
\begin{align}
    \tr_{BB^o}(Z_{AA^oBB^o}) &= \tr_B(W_{AB}) \otimes \psi_{A^o} \label{eq29}
\end{align}
A similar argument on Bob's side yeilds,
\begin{align}
    \tr_{AA^o}(Z_{AA^oBB^o}) &= \tr_A(W_{AB}) \otimes \psi_{B^o} \label{eq30}
\end{align}
Note that it is essential for the states $\psi_{A^o}$ and $\psi_{B^o}$ to be pure in order for equations (\ref{eq29}) and (\ref{eq30}) to hold. From equations (\ref{eq24}), (\ref{eq29}), and (\ref{eq30}) we conclude that 
\begin{align*}
    Z_{AA^oBB^o} = \psi_{A^o} \otimes W_{AB} \otimes \psi_{B_o}
\end{align*}
which is the required expression.

We mention here that the tensor product of a BQS and an entangled state may not necessarily be a valid POPT, as we shown below.

\begin{proposition}\label{propA2}
[Barnum {\it et al.};
\href{https://arxiv.org/abs/quant-ph/0507108}{arXiv:quant-ph/0507108}] For every BQS $W_{A_1\cdots A_n}$, there exists an entangled state $\rho^{s_1 \cdots s_n}_{A_1^o \cdots A_n^o}$ such that the tensor product state $W_{A_1\cdots A_n} \otimes \rho^{s_1 \cdots s_n }_{A_1^o \cdots A_n^o} $ is not a POPT sate.
\end{proposition}
\begin{proof}
Given a BQS $W_{A_1\cdots A_n}$, let $\chi^{s_1 \cdots s_n }_{A_1^o \cdots A_n^o}$ be the eigen-projector corresponding to the negative eigenvalue in the spectral decomposition of $W_{A_1\cdots A_n}$. From Eq. (\ref{projector_basis_expansion}), $\chi_{A_1\cdots A_n}^\mathrm{T}= \sum_{s_1,\cdots,s_n}\beta_{s_1 \cdots s_n} \bigotimes_{i=1}^n {\psi^{s_i}_{A_i}}$. For the local measurements $\left\{P^{+}_{A_iA_i^o}, ~\mathbb{I}_{A_iA_i^o}-P^{+}_{A_iA_i^o}\right\}$ defined in Theorem 2, the probability of the occurrence of outcome $\ P^{+}_{A_1 A_1^o} \otimes \cdots \otimes P^{+}_{A_nA_n^o}$ for the tensor product state $ W_{A_1\cdots A_n} \otimes (\chi^{s_1 \cdots s_n}_{A_1^o \cdots A_n^o})^\mathrm{T} $ is, 
\begin{align}
&\tr\left[\left( P^{+}_{A_1 A_1^o} \otimes \cdots \otimes P^{+}_{A_nA_n^o} \right) \left( W_{A_1\cdots A_n} \otimes (\chi^{s_1 \cdots s_n}_{A_1^o \cdots A_n^o})^\mathrm{T} \right)\right] \nonumber\\ 
&= \sum_{s_1,\cdots,s_n}\beta_{s_1 \cdots s_n}\times\tr\left[\left( P^{+}_{A_1 A_1^o} \otimes \cdots \otimes P^{+}_{A_nA_n^o} \right) \left( W_{A_1\cdots A_n} \otimes \psi^{s_1}_{A_1^o} \otimes \cdots \otimes \psi^{s_n}_{A_n^o} \right)\right] \label{same_equation}\\
&< 0 \nonumber
\end{align}
The final inequality follows from the fact that Eq.(\ref{same_equation}) is the same as Eq.(\ref{same_eqn}), where it is shown that the RHS is negative. Since $W_{A_1\cdots A_n} \otimes (\chi^{s_1 \cdots s_n}_{A_1^o \cdots A_n^o})^\mathrm{T} $ gives negative probabilities upon local measurement, it is not a valid POPT state. Therefore, $(\chi^{s_1 \cdots s_n}_{A_1^o \cdots A_n^o})^\mathrm{T}$ is the required state  $\rho^{s_1 \cdots s_n}_{A_1^o \cdots A_n^o}$ for the proof of our Proposition.
\end{proof}

\begin{figure}[t!]
\centering
\includegraphics[width=0.4\textwidth]{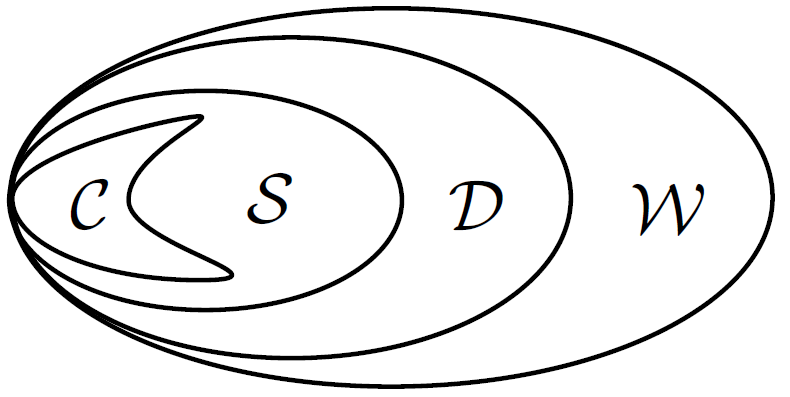}
\caption{Tensoring of POPT states. $\mathcal{W}$ is the set of all POPT states, $\mathcal{D}$ is the set of density operators, $\mathcal{S}$ is the set of separable states, and $\mathcal{C}$ is the set of all classical-classical states having zero discord \cite{Ollivier01}. While all the states in $\mathcal{W}$ allow local quantum description, interesting scenarios arise when tensoring of such states is considered [see Propositions (\ref{propA1}) and (\ref{propA2})]. 
While Barnum {\it et al.} considered the restricted scenario in which states were of the form $w\otimes c$ with $w\in\mathcal{W}$ and $c\in\mathcal{C}$, we have considered more general states of the form  $w\otimes\phi$ with $w\in\mathcal{W}$ and $\phi\in\mathcal{S}$.}\label{fig2}
\end{figure}

\section{Explicit construction of semiquantum game}
{\bf Special case:} Here we first construct a semiquantum game for the BQSs of form $W_p:=p\Gamma[\ket{\phi^+}\bra{\phi^+}]+(1-p)\mathbb{I}/4$. Clearly $W_p$ corresponds to a BQS if and only if $1/3<p\le1$. The entangled state $\ket{\psi^-}=(\ket{01}-\ket{10})/\sqrt{2}\in\mathbb{C}^2\otimes\mathbb{C}^2$ acts as a (beyond quantum) witness for this class of states. This evidently follows from the expression:
\begin{align}
W_p=&\frac{p}{2}\left[\ketbra{0}{0}\otimes\ketbra{0}{0}+\ketbra{1}{1}\otimes\ketbra{1}{1}+\ketbra{\psi^+}{\psi^+}-\ketbra{\psi^-}{\psi^-}\right]+\frac{1-p}{4}\mathbb{I}.
\end{align}
Now the state $\ket{\psi^-}$ allows the following decomposition:
\begin{align}
\ketbra{\psi^-}{\psi^-}&=\frac{1}{2}\left[P_{z}^\mathrm{T} P_{\bar{z}}^\mathrm{T}+P_{\bar{z}}^\mathrm{T} P_{z}^\mathrm{T}-\frac{1}{2}\left(P_{x}^\mathrm{T} P_{x}^\mathrm{T}-P_{\bar{x}}^\mathrm{T} P_{x}^\mathrm{T}-P_{x}^\mathrm{T} P_{\bar{x}}^\mathrm{T}+P_{\bar{x}}^\mathrm{T} P_{\bar{x}}^\mathrm{T}\right.\right.\nonumber\\
&~~~~~~~~~~~~~~~~~~~~~~~~~~~~~~~~~~~~~~~~~~\left.\left.+P_{y}^\mathrm{T} P_{y}^\mathrm{T}-P_{\bar{y}}^\mathrm{T} P_{y}^\mathrm{T}-P_{y}^\mathrm{T} P_{\bar{y}}^\mathrm{T}+P_{\bar{y}}^\mathrm{T} P_{\bar{y}}^\mathrm{T}\right) \right];
\label{q1}
\end{align}
where, $P_iP_j:=P_i\otimes P_j$ with $P_i$ being the projector onto the up eigenstate of $\sigma_i$ for $i\in\left\{x,y,z\right\}$ and it is the projector onto the down eigenstate for $i\in\left\{\bar{x},\bar{y},\bar{z}\right\}$. This immediately leads us to the required semiquantum game $\mathbb{G}_{sq}$. In each run of the game, referee randomly choose the states $\psi^s_{A^o}=P_s$ and $\psi^t_{B^o}=P_t$ and respectively sends them to Alice and Bob without revealing the indices $s$ and t, where $s,t\in\left\{x,y,z,\bar{x},\bar{y},\bar{z}\right\}$. Alice and Bob needs to return classical output $1$ to the referee and the average payoff will be calculated as $\mathcal{I}_{\mathbb{G}_{sq}}:=\sum_{s,t}\beta_{s,t}p(11|\psi^s_{A^o}\psi^t_{B^o})$, where 
\begin{align*}
\beta_{x,x}=\beta_{\bar{x},\bar{x}}=\beta_{y,y}&=\beta_{\bar{y},\bar{y}}=-\beta_{\bar{x},x}=-\beta_{x,\bar{x}}=-\beta_{\bar{y},y}=-\beta_{y,\bar{y}}=\frac{1}{4},\\
\beta_{z,\bar{z}}&=\beta_{\bar{z},z}=\frac{1}{2}~~\mbox{and}~~\beta_{s,t}=0~~\mbox{otherwise}.
\end{align*}
The winning condition demands Alice and Bob to generate a negative payoff. If Alice performs the measurement $\{P^{+}_{AA^o},\mathbb{I}_{AA^o}-P^{+}_{AA^o}\}$ on her part of the shared state $W_{p>1/3}$ and the quantum input $\psi^s_{A^o}$ received from the referee and if Bob also performs the same measurement $\{P^{+}_{BB^o},\mathbb{I}_{BB^o}-P^{+}_{BB^o}\}$ and send the outcome $1$ when the projector $P^{+}_{AA^o}/P^{+}_{BB^o}$ click, then we have
\begin{align*}
\mathcal{I}_{\mathbb{G}_{sq}}\left(W_p\right)
&=\frac{1}{d_Bd_A}\tr\left[\chi_{AB}W_p\right]=\frac{1}{4}\tr\left[\ketbra{\psi^-}{\psi^-}W_p\right]=\frac{1-3p}{16}<0~~\mbox{whenever}~p>1/3.
\end{align*}
On the other hand, for every quantum strategy $\mathcal{I}_{\mathbb{G}_{sq}}(\rho)\ge0$. It should be noted that the decomposition in Eq.(\ref{q1}) is not a unique. Considering a different decomposition it is possible to come up with a different semiquantum game. For instance, one has $\ketbra{\psi^-}{\psi^-}=\sum_{a,b=1}^4\beta_{ab}\psi^a\otimes \psi^b$ where $\{\beta_{ab}\}$ written in a matrix form are given by
\begin{align*}
[\beta_{ab}]:=\begin{bmatrix}
-15/64& ~17/64 & 1/2 & -1/32\\
~17/64& -15/64 & 1/2 & -1/32\\
~1/2  & ~1/2 & -1 & 0\\
-1/32 & -1/32 & 0 & 1/16
\end{bmatrix},
\end{align*}
and $\psi^1:=P_z,~\psi^2:=P_{\bar{z}},~\psi^3:=P_{x}~\mbox{and}~\psi^4:=P_{y}$. 

{\bf General Case:} We now provide an explicit procedure to construct a semiquantum game for any BQS $W_{A_1 \cdots A_n}\in\mathcal{W}\left(\otimes_i\mathcal{H}_{A_i}\right)$.

First note that for a $d$-dimensional Hilbert space $\mathbb{C}^d$ given an orthonormal basis $\left\{\ket{a}\right\}_{a=0}^{d-1}\subset\mathbb{C}^d$, one can construct a non-orthogonal operator basis ($\mathcal{B}^{proj}$) of Projectors from the orthogonal operator (computational) basis ($\mathcal{B}^{comp}$) as follows:
\begin{align*}
\mathcal{L}\left(\mathbb{C}^d\right)\supset\mathcal{B}^{comp} &:=\left\{\ketbra{a}{b}\right\}_{a,b=0}^{d-1};\\
\mathcal{L}\left(\mathbb{C}^d\right)\supset\mathcal{B}^{proj} &:= \{\ketbra{a}{a}\}_{a=0}^{~d-1}\bigcup \{P_1^{a,b},P_2^{a,b}\}_{a,b=0}^{d-1},~~a<b;
\end{align*}
where $P_1^{a,b} := \frac{1}{2} \left( \ketbra{a}{a} + \ketbra{b}{b} + \ketbra{a}{b} + \ketbra{b}{a} \right),~\&~ P_2^{a,b} := \frac{1}{2} \left( \ketbra{a}{a} + \ketbra{b}{b} + i\ketbra{a}{b} -i \ketbra{b}{a} \right)$. Notice that $\mathcal{B}^{proj}$ has $d^2$ linearly independent projectors with $d$ number of them common to $\mathcal{B}^{comp}$. If an operator is known in the $\mathcal{B}^{comp}$ basis then it can be easily written in the $\mathcal{B}^{proj}$ basis by making the following substitution:
\begin{align}
\ketbra{a}{b} =\begin{cases} 
P_1^{a,b}- i P_2^{a,b}-\frac{1-i}{2} \ketbra{a}{a}-\frac{1-i}{2} \ketbra{b}{b},~~ a< b;\\
\\
P_1^{b,a}+ i P_2^{b,a}- \frac{1+i}{2} \ketbra{a}{a} -\frac{1+i}{2} \ketbra{b}{b},~~ b< a. \label{comp_to_proj}
\end{cases}
\end{align}
Now, given an arbitrary beyond quantum state $W_{A_1 \cdots A_n}$, a semi quantum game can be constructed by mimicking the following steps: 
\begin{itemize}
\item[S1:] Write down the spectral decomposition of $W_{A_1 \cdots A_n}$. Hermiticity of $W_{A_1 \cdots A_n}$ guarantees that the eigenvalues are real. Since $W_{A_1 \cdots A_n}$ is a BQS, it has least one negative eigenvalue with entangled eigen-projector. Let the eigen-projector corresponding to a negative eigenvalue ($\lambda<0$) be $\chi_{A_1 \cdots A_n}$. Clearly,
\begin{align}
\tr[W_{A_1\cdots A_n}\chi_{A_1\cdots A_n}]& = \lambda<0, \nonumber\\
\tr[\sigma_{A_1\cdots A_n}\chi_{A_1\cdots A_n}]&\geq 0,~ \forall~\sigma_{A_1\cdots A_n}\in\mathcal{D}(\bigotimes_{i=1}^n\mathcal{H}_{A_i})\nonumber.
\end{align}
\item[S2:] Expand $\chi_{A_1 \cdots A_n}$ in the computational basis:
\begin{align*}
\chi_{A_1 \cdots A_n}= \sum_{a_i^1,\cdots,a_i^n,a_j^1,\cdots,a_j^n}\alpha_{a_i^1 \cdots a_i^n a_j^1 \cdots a_j^n} \ket{a_i^1}_{A_1}\bra{a_j^1} \otimes \cdots \otimes \ket{a_i^n}_{A_n}\bra{a_j^n}.
\end{align*}
Using Eq. (\ref{comp_to_proj}) we can write this as, 
\begin{align*} 
\chi_{A_1\cdots A_n}= \sum_{s_1,\cdots,s_n}\beta_{s_1 \cdots s_n} \bigotimes_{i=1}^n {\phi^{s_i}_{A_i}}, 
\end{align*}
where, ${\phi^{s_i}_{A_i}} \in \mathcal{B}^{proj}_{A_i}$. Since $\chi_{A_1 \cdots A_n}$ is Hermitian and ${\phi^{s_i}_{A_i}}$ are linearly independent, all the $\beta_{s_1 \cdots s_n}$ are real. Let ${\psi^{s_i}_{A_i}} \equiv {\phi^{s_i}_{A_i}}^\mathrm{T}$, where the transpose is taken in the computational basis. 
\begin{align}\label{projector_basis_expansion}
\chi_{A_1\cdots A_n}= \sum_{s_1,\cdots,s_n}\beta_{s_1 \cdots s_n} \bigotimes_{i=1}^n (\psi^{s_i}_{A_i})^\mathrm{T} 
\end{align}
In the semiquantum game, the referee gives one of the pure states ${\psi^{s_i}_{A_i^o}} \in \mathcal{B}^{proj^\mathrm{~T}}_{A_i^o}$ to the $i^{th}$ party in each run (See the proof of Theorem 2
).
\end{itemize}

\section{Proof of Theorem ${\bf 2}$} \label{proof_theo2}
\begin{proof}
This proof is a straightforward generalization of the proof of Theorem ${\bf 1}$. For every BQS $W_{A_1\cdots A_n}\in\mathcal{W}(\bigotimes_{i=1}^n\mathcal{H}_{A_i})$ there exists an entangled state $\chi_{A_1\cdots A_n}\in\mathcal{D}(\bigotimes_{i=1}^n\mathcal{H}_{A_i})$ such that $\tr[W_{A_1\cdots A_n}\chi_{A_1\cdots A_n}]<0$, whereas $\tr[\sigma_{A_1\cdots A_n}\chi_{A_1\cdots A_n}]\geq 0$, $\forall~\sigma_{A_1\cdots A_n}\in\mathcal{D}(\bigotimes_{i=1}^n\mathcal{H}_{A_i})$ \cite{Guhne09}. The state allows non-unique decomposition of the form
\begin{align}
\chi_{A_1\cdots A_n}= \sum_{s_1,\cdots,s_n}\beta_{s_1 \cdots s_n} \bigotimes_{i=1}^n {\psi^{s_i}_{A_i}}^\mathrm{T},~~
\mbox{where~}\psi^{s_i}_{A_i} \in\mathcal{D}(\mathcal{H}_{A_i})~~\&~~\beta_{s_1 \cdots s_n}\in\mathbb{R}. \nonumber
\end{align}
In the semiquantum game, referee sends the quantum inputs $\psi^{s_i}_{A_i^o}$ to the $i^{th}$ party who has to produce binary outputs $\in \{0,1\}$. Their average payoff will be calculated as 
\begin{align}
\mathcal{I}_{\mathbb{G}^\chi_{sq}}:=\sum_{s_1,\cdots,s_n}\beta_{s_1 \cdots s_n}\times p\left(1\cdots1|\psi^{s_1}_{A_1^o}\cdots\psi^{s_N}_{A_n^o}\right). \nonumber
\end{align}
Given the BQS $W_{A_1\cdots A_n}$, the $i^{th}$ party performs the measurement $\left\{P^{+}_{A_iA_i^o}, ~\mathbb{I}_{A_iA_i^o}-P^{+}_{A_iA_i^o}\right\}$ on her part of the shared BQS and the quantum input $\psi^{s_i}_{A_i^o}$ received from the referee.  Here $P_{A_i A_i^o}^{+}:=\ket{\phi^+}_{A_i A_i^o}\bra{\phi^+}$ with $\ket{\phi^+}_{A_i A_i^o}:=\frac{1}{\sqrt{d_{A_i}}}\sum^{d_{A_i}-1}_{i=0}|ii\rangle$ and $P^{+}_{A_iA_i^o}$ corresponds to the output $1$. The average payoff turns out to be, 
\begin{align}
\mathcal{I}_{\mathbb{G}^\chi_{sq}}\left(W_{A_1\cdots A_n}\right)&=\sum_{s_1,\cdots,s_n}\beta_{s_1 \cdots s_n}\times\tr\left[\left( P^{+}_{A_1 A_1^o} \otimes \cdots \otimes P^{+}_{A_nA_n^o} \right) \left( W_{A_1\cdots A_n} \otimes \psi^{s_1}_{A_1^o} \otimes \cdots \otimes \psi^{s_n}_{A_n^o} \right)\right] \label{same_eqn}\\
&=\sum_{s_1,\cdots,s_n}\beta_{s_1 \cdots s_n}\times\tr\left[\left( R_{A_1} \otimes  \cdots \otimes R_{A_n} \right)W_{A_1\cdots A_n}\right],\nonumber
\end{align} 
where $R_{A_i}$ is the effective POVM acting on the $i^{th}$ party's part of $W_{A_1\cdots A_n}$, and is given by
$R_{A_i}:=\tr_{A_i^o}\left[P^+_{A_iA_i^o}\left(\mathbb{I}_{A_i}\otimes \psi^{s_i}_{A_i^o}\right)\right]=\frac{1}{d_{A_i}}{\psi^{s_i}_{A_i}}^\mathrm{T}$. Therefore, we have,
\begin{align}
\mathcal{I}_{\mathbb{G}^\chi_{sq}}\left(W_{A_1\cdots A_n}\right)&=\prod_{i=1}^nd^{-1}_{A_i}\sum_{s_1,\cdots,s_n}\beta_{s_1 \cdots s_n}\times\tr\left[\left(\bigotimes_{i=1}^n {\psi^{s_i}_{A_i}}^\mathrm{T}\right)W_{A_1\cdots A_n}\right]\nonumber\\
&=\prod_{i=1}^nd^{-1}_{A_i}\tr\left[\left(\sum_{s_1,\cdots,s_n}\beta_{s_1 \cdots s_n} \bigotimes_{i=1}^n {\psi^{s_i}_{A_i}}^\mathrm{T} \right)W_{A_1\cdots A_n}\right]\nonumber\\
&=\prod_{i=1}^nd^{-1}_{A_i}\tr\left[\chi_{A_1\cdots A_n}W_{A_1\cdots A_n}\right]<0. \nonumber
\end{align}
We will now calculate the payoff for an arbitrary quantum strategy. Given a quantum state $\rho_{A_1 \cdots A_N}$ let the $i^{th}$ party perform the measurement $M_{A_iA_i^o}\equiv\{\pi^{a_i}_{A_iA_i^o}\}$ on her respective joint system, where $a_i \in\{0,1\}$. The average payoff turns out to be 
\begin{align}
\mathcal{I}_{\mathbb{G}^\chi_{sq}}\left(\rho_{A_1 \cdots A_n}\right)&=\sum_{s_1,\cdots,s_n}\beta_{s_1 \cdots s_n}\times\tr\left[ \left(\pi^1_{A_1A_1^o} \otimes \cdots \otimes \pi^1_{A_nA_n^o} \right)\left(\rho_{A_1 \cdots A_n} \otimes \psi^{s_1}_{A_1^o} \otimes \cdots \otimes \psi^{s_n}_{A_n^o} \right)\right]\nonumber\\
&=\sum_{s_1,\cdots,s_n}\beta_{s_1 \cdots s_n}\times\tr\left[R_{A_1^o \cdots A_n^o}\left(\bigotimes_{i=1}^n \psi^{s_i}_{A_i^o}\right)\right],\nonumber
\end{align} 
where, $R_{A_1^o \cdots A_n^o}:=\tr_{A_1 \cdots A_n}\left[\left(\pi^1_{A_1A_1^o} \otimes \cdots \otimes \pi^1_{A_n A_n^o}\right)\left(\rho_{A_1 \cdots A_n} \otimes \mathbb{I}_{A_1^o \cdots A_n^o}\right)\right]$
is a positive semidefinite operator, {\it i.e.} $R_{A_1^o \cdots A_n^o}\in\mathcal{E}\left(\bigotimes_{i=1}^n\mathcal{H}_{A_i^o}\right)$. Linearity of trace further yields,
\begin{align}
\mathcal{I}_{\mathbb{G}^\chi_{sq}}\left(\rho_{A_1 \cdots A_N}\right)&=\tr\left[R_{A_1^o \cdots A_n^o}\left(\sum_{s_1,\cdots,s_n}\beta_{s_1 \cdots s_N}\bigotimes_{i=1}^n\psi^{s_i}_{A_i^o}\right)\right]=\tr\left[R_{A_1^o \cdots A_n^o}~\chi^\mathrm{T}_{A_1^o \cdots A_n^o}\right]\geq0.
\end{align}
The last inequality follows from the fact that $\chi^\mathrm{T}_{A_1^o \cdots A_n^o}\in\mathcal{D}\left(\bigotimes_{i=1}^n\mathcal{H}_{A_i^o}\right)$, and this completes the proof.
\end{proof}

\section{Necessity of non-orthogonal inputs in Theorem ${\bf 1}$ and Theorem ${\bf 2}$}
In this section we will discuss the necessity of the non-orthogonal quantum inputs in the game $\mathbb{G}_{sq}$ used in Theorem ${\bf 1}$ and Theorem ${\bf 2}$. While paying the game $\mathbb{G}_{sq}$, let the $i^{th}$ party get the quantum input $\psi^{s_i}_{A^o_i} \in\mathcal{D}(\mathcal{H}_{A^o_i})$ and perform some joint measurement $M_{A_i^o A_i}\equiv \{\pi^{a_i}_{A_i^o A_i}\}$, where $a_i$ is the outcome corresponding to the POVM element $\pi^{a_i}_{A_i^o A_i}$. For the BQS $W_{A_1\cdots A_n}$, the joint probabilities are given by 
\begin{align}
p(a_1,\cdots,a_n|\psi^{s_1}_{A^o_1},\cdots,\psi^{s_n}_{A^o_n})&= \tr\left[\bigotimes_i\pi^{a_i}_{A_i^o A_i}\left(\bigotimes_i\psi^{s_i}_{A^o_i}\otimes W_{A_1\cdots A_n}\right)\right]\nonumber\\
&\equiv\tr\left[\left(\bigotimes_i Q^{a_i}_{A_i}[s_i]\right) W_{A_1\cdots A_n}\right],
\end{align}
where, $Q^{a_i}_{A_i}[s_i]:=\tr_{A_i^o}\left[\pi^{a_i}_{A_iA^o_i}(I_{A_i}\otimes \psi^{s_i}_{A^o_i})\right]\in\mathcal{E}\left(\mathcal{H}_{A_i}\right)$  effectively acts on $A_i$ subsystem of the shared state $W_{A_1\cdots A_n}$ when the quantum input $\psi^{s_i}_{A^o_i}$ is given by the referee. Since $\sum_{a_i}\pi^{a_i}_{A_i A^o_i}=\mathbb{I}_{A_iA^o_i}$, we have,
\begin{align}
\sum_{a_i}Q^{a_i}_{A_i}[s_i]=\tr_{A_i^o}\left[\left(\sum_{a_i}\pi^{a_i}_{A_iA^o_i}\right)\left(\mathbb{I}_{A_i}\otimes \psi^{s_i}_{A^o_i}\right)\right]=\tr_{A_i^o}\left[\mathbb{I}_{A_i}\otimes \psi^{s_i}_{A^o_i}\right]=\mathbb{I}_{A_i}.
\end{align}
Therefore, $M^{s_i}_{A_i}\equiv\left\{Q^{a_i}_{A_i}[s_i]\right\}$ is the effective measurement performed by the $A_i^{th}$ party on the shared state $W_{A_1\cdots A_n}$ when the quantum input $\psi^{s_i}_{A^o_i}$ is received by the $i^{th}$ party.

Let us now assume that the BQS is of the following form
\begin{align}
W_{A_1\cdots A_n}=\sum_kp_k\left(\Lambda^k_{A_1}\otimes \cdots \otimes \Lambda^k_{A_n}\right)\rho^k \label{W-form}
\end{align} 
where, $\rho^k\in \mathcal{D}(\otimes_i\mathcal{H}_{A_i})$, $\Lambda_i^k$ are positive trace-preserving maps, and $\{p_k\}$ is a probability distribution. In this case we have,
\begin{align}
p\left(a_1,\cdots,a_n|\psi^{s_1}_{A^o_1},\cdots,\psi^{s_n}_{A^o_n}\right)&=\tr\left[\left(\bigotimes_i Q^{a_i}_{A_i}[s_i]\right) \sum_kp_k\left(\bigotimes_i\Lambda^k_{A_i}\right)\rho^k\right]\nonumber\\
&=\sum_kp_k\tr\left[\left(\bigotimes_i Q^{a_i}_{A_i}[s_i]\right)\left(\bigotimes_i\Lambda^k_{A_i}\right)\rho^k\right]\nonumber\\
&=\sum_kp_k\tr\left[\left\{\bigotimes_i\Lambda_{A_i}^{*k}\left(Q^{a_i}_{A_i}[s_i]\right)\right\}\rho^k\right]\nonumber\\
&=\sum_kp_k\tr\left[\left(\bigotimes_i \tilde{Q}^{a_i}_{A_i}[s_i]\right)\rho^k\right] ,
\end{align}
where $\Lambda^*$ is the adjoint map of $\Lambda$, {\it i.e.} $\tr[U\Lambda(V)]=\tr[\Lambda^*(U)V]$ for all Hermitial matrices $U$ and $V$. Clearly $\tilde{M}^{s_i}_{A_i}\equiv \{\tilde{Q}^{a_i}_{A_i}[s_i]\}$ is a valid quantum measurement since the dual of a positive trace-preserving map is positive and unital. Therefore, for the class of BQSs given by Eq.(\ref{W-form}), whenever the input states are orthogonal, the correlations generated by the BQS can by simulated quantum mechanically as follows: \\ 
The $i^{th}$ party first performs a measurement to identify the index `$s_i$' of the given quantum state $\psi^{s_i}_{A^o_i}$ and then performs the measurement $\tilde{M}^{s_i}_{A_i}\equiv \{\tilde{Q}^{a_i}_{A_i}[s_i]\}$ on her part of the multipartite quantum state $\rho^k$. This generates the correlation $p(a_1,\cdots,a_n|\psi^{s_1}_{A^o_1},\cdots,\psi^{s_n}_{A^o_n})$ which was obtained by performing the local measurements
$M_{A_i^o A_i}\equiv \{\pi^{a_i}_{A_i^o A_i}\}$ on $\bigotimes_i\psi^{s_i}_{A^o_i}\otimes W_{A_1\cdots A_n}$. Note that if the inputs are orthogonal then the index `$s_i$' can be identified unambiguously. Therefore, when the BQSs are of the form (\ref{W-form}), non-orthogonal inputs are necessary to obtain the advantage of BQS over quantum states. This reproduces the results in \cite{Barnum10,Acin10}.

\twocolumngrid


\begin{thebibliography}{99}
	
\bibitem{Bell64} J.S. Bell; On the Einstein Podolsky Rosen paradox,
\href{https://doi.org/10.1103/PhysicsPhysiqueFizika.1.195}{Physics {\bf 1}, 195 (1964)}.	
	
\bibitem{Bell66} J. S. Bell; On the Problem of Hidden Variables in Quantum Mechanics,
\href{https://doi.org/10.1103/RevModPhys.38.447}{Rev. Mod. Phys. {\bf 38}, 447 (1966)}.	

\bibitem{Mermin93} N. D. Mermin; Hidden variables and the two theorems of John Bell,
\href{https://doi.org/10.1103/RevModPhys.65.803}{Rev. Mod. Phys. {\bf 65}, 803 (1993)}.

\bibitem{Brunner14} N. Brunner, D. Cavalcanti, S. Pironio, V. Scarani, and S. Wehner; Bell nonlocality,
\href{https://doi.org/10.1103/RevModPhys.86.419}{Rev. Mod. Phys. {\bf 86}, 419 (2014)}.

\bibitem{Schrodinger35} E. Schr\"{o}dinger; Discussion of Probability Relations Between Separated Systems, 
\href{https://doi.org/10.1017/s0305004100013554}{Math. Proc. Camb. Philos. Soc {\bf 31}, 555 (1935)}.

\bibitem{Barrett05} J. Barrett, L. Hardy, and A. Kent; No signalling and Quantum Key Distribution,
\href{https://doi.org/10.1103/PhysRevLett.95.010503}{Phys. Rev. Lett. {\bf 95}, 010503 (2005)};
\bibitem{Acin06} A. Acín, N. Gisin, and L. Masanes; From Bell’s Theorem to Secure Quantum Key Distribution,
\href{https://doi.org/10.1103/PhysRevLett.97.120405}{Phys. Rev. Lett. {\bf 97}, 120405 (2006)};
\bibitem{Acin07} A. Acín, N. Brunner, N. Gisin, S. Massar, S. Pironio, and V. Scarani; Device-Independent Security of Quantum Cryptography against Collective Attacks, 
\href{https://doi.org/10.1103/PhysRevLett.98.230501}{Phys. Rev. Lett. {\bf 98}, 230501 (2007)}.

\bibitem{Pironio10} S. Pironio {\it et al.} Random numbers certified by Bell’s theorem,
\href{https://doi.org/10.1038/nature09008}{Nature {\bf 464}, 1021 (2010)}. 
	
\bibitem{Colbeck12} R. Colbeck and R. Renner; Free randomness can be amplified,
\href{https://doi.org/10.1038/nphys2300}{Nature Phys {\bf 8}, 450 (2012)}. 

\bibitem{Brunner08} N. Brunner, S. Pironio, A. Acin, N. Gisin, A. A. Méthot, and V. Scarani; Testing the Dimension of Hilbert Spaces,
\href{https://doi.org/10.1103/PhysRevLett.100.210503}{Phys. Rev. Lett. {\bf 100}, 210503 (2008)};

\bibitem{Gallego10} 
R. Gallego, N. Brunner, C. Hadley, and A. Acín; Device-Independent Tests of Classical and Quantum Dimensions,
\href{https://doi.org/10.1103/PhysRevLett.105.230501}{Phys. Rev. Lett. {\bf 105}, 230501 (2010)}.

\bibitem{Mukherjee15} 
A. Mukherjee, A. Roy, S. S. Bhattacharya, S. Das, Md. R. Gazi, and M. Banik; Hardy's test as a device-independent dimension witness,
\href{https://doi.org/10.1103/PhysRevA.92.022302}{Phys. Rev. A {\bf 92}, 022302 (2015)}.

\bibitem{Cirelson80} B. S. Cirel'son; Quantum generalizations of Bell's inequality,
\href{https://doi.org/10.1007/BF00417500}{Lett. Math. Phys. {\bf 4}, 93 (1980)}. 

\bibitem{Popescu94} S. Popescu and D. Rohrlich; Quantum nonlocality as an axiom,
\href{https://doi.org/10.1007/BF02058098}{Found. Phys. {\bf 24}, 379 (1994)}.

\bibitem{Clauser69} J. F. Clauser, M. A. Horne, A. Shimony, and R. A. Holt; Proposed Experiment to Test Local Hidden-Variable Theories,
\href{https://doi.org/10.1103/PhysRevLett.23.880}{Phys. Rev. Lett. {\bf 23}, 880 (1969)}.

\bibitem{Birkhoff36} G. Birkhoff and J. von Neumann; The logic of quantum mechanics, 
\href{https://doi.org/10.2307/1968621}{Ann. Math. {\bf 37}, 823 (1936)};
\bibitem{Mackey63} 
G. W. Mackey; Mathematical Foundations of Quantum Mechanics. Benjamin, W. A. New York, 1963; Dover reprint, 2004;
\bibitem{Ludwi68} 
G. Ludwig; Attempt of an axiomatic foundation of quantum mechanics and more general theories II, III,
\href{https://doi.org/10.1007/BF01653647}{Commun. Math. Phys. {\bf 4}, 331 (1967)}; \href{https://doi.org/10.1007/BF01654027}{Commun. Math. Phys. {\bf 9}, 1 (1968)};
\bibitem{Mielni68}  
B. Mielnik; Geometry of quantum states, 
\href{https://doi.org/10.1007/BF01654032}{Commun. Math. Phys. {\bf 9}, 55 (1968)};
\bibitem{Beltramett81} 
E. Beltrametti and G. Cassinelli; The Logic of Quantum Mechanics, Addison-Wesley (1981);
\bibitem{Soler95} 
M. P. Sol{\`{e}}r; Characterization of hilbert spaces by orthomodular spaces, 
\href{https://doi.org/10.1080/00927879508825218}{Commun. Algebra {\bf 23}, 219 (1995)};
\bibitem{Haag96} 
R. Haag; Local Quantum Physics: Fields, Particles, Algebras,
2nd Revised and Enlarged version, Springer (1996);
\bibitem{Clifton03} 
R. Clifton,J. Bub, and H. Halvorson; Characterizing Quantum Theory in Terms of Information-Theoretic Constraints,
\href{https://doi.org/10.1023/A:1026056716397}{Found. Phys. {\bf 33}, 1561 (2003)};
\bibitem{Barrett07} 
J. Barrett; Information processing in generalized probabilistic theories,
\href{https://doi.org/10.1103/PhysRevA.75.032304}{Phys. Rev. A {\bf 75}, 032304 (2007)};
\bibitem{Abramsky08} 
S. Abramsky and B. Coecke; Categorical quantum mechanics, Handbook of Quantum
Logic and Quantum Structures vol II, Elsevier, Amsterdam (2008);
\bibitem{Chiribella11} 
G. Chiribella, G. Mauro D’Ariano, and P. Perinotti; Informational derivation of quantum theory,
\href{https://doi.org/10.1103/PhysRevA.84.012311}{Phys. Rev. A {\bf 84}, 012311 (2011)}.

\bibitem{Foulis80} D. Foulis and C. Randall, in Interpretations and Foundations of Quantum Theory, edited by H. Neumann
(Bibliographisches Institut Wissenschaftverlag, Mannheim, 1980), Vol. 5, pp. 9–20.

\bibitem{Klay87} M. Kläy, C. Randall, and D. Foulis; Tensor Products and Probability Weights, \href{https://doi.org/10.1007/bf00668911}{Int. J. Theor. Phys. {\bf 26}, 199 (1987)}.

\bibitem{Wallach00} N. R. Wallach; An Unentangled Gleason's Theorem,
\href{https://arxiv.org/abs/quant-ph/0002058}{arXiv:quant-ph/0002058}.

\bibitem{Barnum05} H. Barnum, C. A. Fuchs, J. M. Renes, and A. Wilce; Influence-free states on compound quantum systems,
\href{https://arxiv.org/abs/quant-ph/0507108}{arXiv:quant-ph/0507108}.

\bibitem{Barnum10} H. Barnum, S. Beigi, S. Boixo, M. B. Elliott, and S. Wehner; Local Quantum Measurement and No-Signaling Imply Quantum Correlations, \href{https://doi.org/10.1103/PhysRevLett.104.140401}{Phys. Rev. Lett. {\bf 104}, 140401 (2010)}.

\bibitem{Acin10} A. Acín, R. Augusiak, D. Cavalcanti, C. Hadley, J. K. Korbicz, M. Lewenstein, Ll. Masanes, and M. Piani; Unified Framework for Correlations in Terms of Local Quantum Observables,
\href{https://doi.org/10.1103/PhysRevLett.104.140404}{Phys. Rev. Lett. {\bf 104}, 140404 (2010)}.

\bibitem{Torre12} G. de la Torre, L. Masanes, A. J. Short, and M. P. Müller; Deriving Quantum Theory from Its Local Structure and Reversibility,
\href{https://doi.org/10.1103/PhysRevLett.109.090403}{Phys. Rev. Lett. {\bf 109}, 090403 (2012)}.

\bibitem{Kleinmann13} M. Kleinmann, T. J. Osborne, V. B. Scholz, and A. H. Werner; Typical Local Measurements in Generalized Probabilistic Theories: Emergence of Quantum Bipartite Correlations,
\href{https://doi.org/10.1103/PhysRevLett.110.040403}{Phys. Rev. Lett. {\bf 110}, 040403 (2013)}.

\bibitem{Gleason57} A. M. Gleason; Measures on the Closed Subspaces of a Hilbert Space,
\href{https://doi.org/10.1512/iumj.1957.6.56050}{J. Math. Mech. {\bf 6}, 885 (1957)}.

\bibitem{Busch03} P. Busch; Quantum States and Generalized Observables: A Simple Proof of Gleason’s Theorem,
\href{https://doi.org/10.1103/PhysRevLett.91.120403}{Phys. Rev. Lett. {\bf 91}, 120403 (2003)}.

\bibitem{Caves04} C. M. Caves, C. A. Fuchs, K. K. Manne, and J. M. Renes; Gleason-Type Derivations of the Quantum Probability Rule for Generalized Measurements, \href{https://doi.org/10.1023/B:FOOP.0000019581.00318.a5}{Found. Phys. {\bf 34}, 193 (2004)}.

\bibitem{Buscemi12} F. Buscemi; All Entangled Quantum States Are Nonlocal,
\href{https://doi.org/10.1103/PhysRevLett.108.200401}{Phys. Rev. Lett. {\bf 108}, 200401 (2012)}. 

\bibitem{Werner89} R. F. Werner; Quantum states with Einstein-Podolsky-Rosen correlations admitting a hidden-variable model,
\href{https://doi.org/10.1103/PhysRevA.40.4277}{Phys. Rev. A {\bf 40}, 4277 (1989)}.

\bibitem{Barrett02} J. Barrett; Nonsequential positive-operator-valued measurements on entangled mixed states do not always violate a Bell inequality,
\href{https://doi.org/10.1103/PhysRevA.65.042302}{Phys. Rev. A {\bf 65}, 042302 (2002)}.

\bibitem{Rai12} A. Rai, MD. R. Gazi, M. Banik, S. Das, and S. Kunkri; Local simulation of singlet statistics for restricted set of measurement,
\href{https://doi.org/10.1088/1751-8113/45/47/475302}{J. Phys. A: Math. Theor. {\bf 45}, 475302 (2012)}.


\bibitem{Kraus83} K. Kraus; States, Effects, and Operations: Fundamental Notions of Quantum Theory, Eds. A. Böhm, J. D. Dollard, and W. H. Wootters, Springer-Verlag Berlin Heidelberg (1983).

\bibitem{Guhne09} O. Gühne and G. Tóth; Entanglement detection,
\href{https://doi.org/10.1016/j.physrep.2009.02.004}{Phys. Rep {\bf 474}, 1 (2009)}.


\bibitem{Branciard13} C. Branciard, D. Rosset, Y-C Liang, and N. Gisin; Measurement-Device-Independent Entanglement Witnesses for All Entangled Quantum States,
\href{https://doi.org/10.1103/PhysRevLett.110.060405}{Phys. Rev. Lett. {\bf 110}, 060405 (2013)}.

\bibitem{Banik13} M. Banik; Lack of measurement independence can simulate quantum correlations even when signalling can not,
\href{https://doi.org/10.1103/PhysRevA.88.032118}{Phys. Rev. A {\bf 88}, 032118 (2013)}.

\bibitem{Chaturvedi15} A. Chaturvedi and M. Banik; Measurement-device–independent randomness from local entangled states,
\href{https://doi.org/10.1209/0295-5075/112/30003}{EPL {\bf 112}, 30003 (2015)}.

\bibitem{Rosset20} D. Rosset, D. Schmid, and F. Buscemi; Type-Independent Characterization of Spacelike Separated Resources,
\href{https://doi.org/10.1103/PhysRevLett.125.210402}{Phys. Rev. Lett. {\bf 125}, 210402 (2020)}.

\bibitem{Schmid20} D. Schmid, D. Rosset, and F. Buscemi; The type-independent resource theory of local operations and shared randomness,
\href{https://doi.org/10.22331/q-2020-04-30-262}{Quantum {\bf 4}, 262 (2020)}.

\bibitem{Graffitti20} F. Graffitti, A. Pickston, P. Barrow, M. Proietti, D. Kundys, D. Rosset, M. Ringbauer, and A. Fedrizzi; Measurement-Device-Independent Verification of Quantum Channels,
\href{https://doi.org/10.1103/PhysRevLett.124.010503}{Phys. Rev. Lett. {\bf 124}, 010503 (2020)}.


\bibitem{Woronowicz76} S. L. Woronowicz; Positive maps of low dimensional matrix algebras,
\href{https://doi.org/10.1016/0034-4877(76)90038-0}{Rep. Math. Phys. {\bf 10}, 165 (1976)}.

\bibitem{Horodecki96} M. Horodecki, P. Horodecki, and R. Horodecki; Separability of mixed states: necessary and sufficient conditions,
\href{https://doi.org/10.1016/S0375-9601(96)00706-2}{Phys. Lett. A {\bf 223}, 1 (1996)}.

\bibitem{Stormer13} E. Størmer; Positive Linear Maps of Operator Algebras, Springer-Verlag Berlin Heidelberg (2013).

\bibitem{vanDam05} W. van Dam; Implausible Consequences of Superstrong Nonlocality,
\href{https://doi.org/10.1007/s11047-012-9353-6}{Nat. Comput. {\bf 12}, 9 (2013)}
[see also \href{https://arxiv.org/abs/quant-ph/0501159}{arXiv:quant-ph/0501159 (2005)}].

\bibitem{Buhrman10} 
H. Buhrman, R. Cleve, S. Massar, and R. de Wolf; Nonlocality and communication complexity,
\href{https://doi.org/10.1103/RevModPhys.82.665}{Rev. Mod. Phys. {\bf 82}, 665 (2010)}.


\bibitem{Pawlowski09} 
M. Pawłowski, T. Paterek, D. Kaszlikowski, V. Scarani, A. Winter, and M. Żukowski; Information causality as a physica principle,
\href{https://doi.org/10.1038/nature08400}{Nature {\bf 461}, 1101 (2009)}.

\bibitem{Navascues09} 
M. Navascues and H. Wunderlich; A glance beyond the quantum model,
\href{https://doi.org/10.1098/rspa.2009.0453}{Proc. Roy. Soc. Lond. A {\bf 466}, 881 (2009)}.

\bibitem{Fritz13} 
T. Fritz, A. B. Sainz, R. Augusiak, J. B. Brask, R. Chaves, A. Leverrier, and A. Acín; Local orthogonality as a multipartite principle for quantum correlations,
\href{https://doi.org/10.1038/ncomms3263}{Nat. Commun. {\bf 4}, 2263 (2013)}.

\bibitem{Cabello13} 
A. Cabello; Simple Explanation of the Quantum Violation of a Fundamental Inequality,
\href{https://doi.org/10.1103/PhysRevLett.110.060402}{Phys. Rev. Lett. {\bf 110}, 060402 (2013)}.

\bibitem{Oppenheim10} 
J. Oppenheim and S. Wehner; The uncertainty principle determines the non-locality of quantum mechanics,
\href{https://doi.org/10.1126/science.1192065}{Science {\bf 330}, 1072 (2010)}.

\bibitem{Banik13(1)} 
M. Banik, Md. R. Gazi, S. Ghosh, and G. Kar; Degree of complementarity determines the nonlocality in quantum mechanics,
\href{https://doi.org/10.1103/PhysRevA.87.052125}{Phys. Rev. A {\bf 87}, 052125 (2013)}.

\bibitem{Banik15} 
M. Banik, S. S. Bhattacharya, A. Mukherjee, A. Roy, A. Ambainis, and A. Rai; Limited preparation contextuality in quantum theory and its relation to the Cirel'son bound,
\href{https://doi.org/10.1103/PhysRevA.92.030103}{Phys. Rev. A {\bf 92}, 030103(R) (2015)}.

\bibitem{Kar16(1)} 
G. Kar and M. Banik; Several foundational and information theo- retic implications of Bell’s theorem; 
\href{https://doi.org/10.1142/S021974991640027X}{Int. J. Quantum Inform. {\bf 14}, 1640027 (2016)}.

\bibitem{Kar16(2)} 
G. Kar, S.Ghosh, S. K. Choudhary, and M. Banik; Role of Measurement Incompatibility and Uncertainty in Determining Nonlocality,
\href{https://doi.org/10.3390/math4030052}{Mathematics {\bf 4}, 52 (2016)}.

\bibitem{Banik19} 
M. Banik, S. Saha, T. Guha, S. Agrawal, S. S. Bhattacharya, A. Roy, and A. S. Majumdar; Constraining the state space in any physical theory with the principle of information symmetry,
\href{https://doi.org/10.1103/PhysRevA.100.060101}{Phys. Rev. A {\bf 100}, 060101(R) (2019)}.

\bibitem{Bhattacharya20} 
S. S. Bhattacharya, S. Saha, T. Guha, and M. Banik; Nonlocality without entanglement: Quantum theory and beyond,
\href{https://doi.org/10.1103/PhysRevResearch.2.012068}{Phys. Rev. Research {\bf 2}, 012068(R) (2020)}.

\bibitem{DallArno17} M. Dall’Arno, S. Brandsen, A. Tosini, F. Buscemi, and V. Vedral; No-Hypersignalling Principle,
\href{https://doi.org/10.1103/PhysRevLett.119.020401}{Phys. Rev. Lett. {\bf 119}, 020401 (2017)}.

\bibitem{Naik21} S. G. Naik, E. P. Lobo, S. Sen, R. Patra, M. Alimuddin, T. Guha, S. S. Bhattacharya, and M. Banik; Composition of multipartite quantum systems: perspective from time-like paradigm,
\href{https://doi.org/10.1103/PhysRevLett.128.140401}{Phys. Rev. Lett. {\bf 128}, 140401 (2022)}.

\bibitem{Sen2022} S. Sen, E. P. Lobo, R. K. Patra, S. G. Naik, A. D. Bhowmik, M. Alimuddin, and M. Banik; Timelike correlations and quantum tensor product structure,
\href{https://doi.org/10.48550/arXiv.2208.02471}{arXiv:2208.02471}.

\bibitem{Popescu94(1)} S. Popescu; Bell’s inequalities versus teleportation: What is nonlocality?
\href{https://doi.org/10.1103/PhysRevLett.72.797}{Phys. Rev. Lett. {\bf 72}, 797 (1994)}.

\bibitem{Ollivier01} H. Ollivier and W. H. Zurek; Quantum Discord: A Measure of the Quantumness of Correlations,
\href{https://doi.org/10.1103/PhysRevLett.88.017901}{Phys. Rev. Lett. {\bf 88}, 017901 (2001)}

\bibitem{Henderson01} L. Henderson and V. Vedral; Classical, quantum and total correlations,
\href{https://doi.org/10.1088/0305-4470/34/35/315}{J. Phys. A: Math. Gen. {\bf 34}, 6899 (2001)}.

\end{thebibliography}
\end{document}